\documentclass[journal]{IEEEtran}
\usepackage{amsmath,amsfonts,amsthm,amssymb}
\usepackage{algorithmic}
\usepackage{algorithm}
\usepackage{enumitem}
\usepackage{array}
\usepackage[caption=false,font=normalsize,labelfont=sf,textfont=sf]{subfig}
\usepackage{textcomp}
\usepackage{stfloats}
\usepackage{multirow}
\usepackage{url}
\usepackage{verbatim}
\usepackage{booktabs}
\usepackage{graphicx}
\usepackage{cite}
\usepackage{xcolor}
\usepackage{cancel}
\usepackage[scr=boondoxo,scrscaled=1.05]{mathalfa}
\usepackage{bm}
\newtheorem{theorem}{Theorem}
\newtheorem{lemma}{Lemma}
\newtheorem{prop}{Proposition}
\newtheorem{remark}{Remark}

\newcommand{\eor}{\ensuremath{\hfill\blacklozenge}}

\providecommand{\norm}[1]{\left\lVert#1\right\rVert}

\newcommand{\J}{{\cal J}}
\newcommand{\bbm}{\begin{bmatrix}}
\newcommand{\ebm}{\end{bmatrix}}

\begin{document}

\title{Dual Unscented Kalman Filter Architecture for Sensor Fusion in Water Networks Leak Localization}

\author{Luis~Romero-Ben,    Paul~Irofti,~\IEEEmembership{Member,~IEEE,}
Florin~Stoican,~\IEEEmembership{Member,~IEEE,}
and~Vicenç~Puig
\thanks{Luis Romero-Ben and Vicenç Puig are with the Institut de Robòtica i Informàtica Industrial, CSIC-UPC,
e-mail: \{luis.romero.ben,vicenc.puig\}@upc.edu}
\thanks{Paul Irofti is with LOS-CS-FMI,
University of Bucharest, 
e-mail: {paul@irofti.net}}%

\thanks{Vicenç Puig is also with the Supervision, Safety and Automatic Control Research Center (CS2AC) of the Universitat Politècnica de Catalunya
}
\thanks{Florin Stoican is with the  Dept. of Automation Control and Systems Engineering, Politehnica University of Bucharest,
e-mail: {florin.stoican@upb.ro}.}
}



\maketitle

\begin{abstract}
Leakage in water systems results in significant daily water losses, degrading service quality, increasing costs, and aggravating environmental problems. Most leak localization methods rely solely on pressure data, missing valuable information from other sensor types. This article proposes a hydraulic state estimation methodology based on a dual Unscented Kalman Filter (UKF) approach, which enhances the estimation of both nodal hydraulic heads, critical in localization tasks, and pipe flows, useful for operational purposes. The approach enables the fusion of different sensor types, such as pressure, flow and demand meters. The strategy is evaluated in well-known open source case studies, namely Modena and L-TOWN, showing improvements over other state-of-the-art estimation approaches in terms of interpolation accuracy, as well as more precise leak localization performance in L-TOWN.
\end{abstract}




\section{Introduction}
\IEEEPARstart{W}{ater} resources are becoming increasingly scarce while global demand continues to rise, with projections estimating a 55\% increase between 2000 and 2050 \cite{Leflaive2012}.  
The efficient delivery of water minimizes losses, but leaks remain a major challenge for water utilities. 
Globally, leakage is estimated to produce losses of around 126 billion cubic meters of water per year \cite{Liemberger2019}. 
Water utilities have a great interest in developing leak detection/localization algorithms to minimize repair times, service disruptions, and water losses in water distribution networks (WDNs). Traditionally, leaks have been detected through night flow analysis and localized using portable acoustic sensors. However, this is inefficient for large-scale systems, as inspection is limited to reduced sections of the WDN at a time. This limitation has driven the demand for software-based solutions that can automatically analyze sensor data using models and algorithms, enabling simultaneous real-time, system-wide monitoring. 

Most leak localization methods are model-based, relying on hydraulic models to simulate the behaviour of the WDN. Simulated hydraulic data is then compared to actual measurements to localize the leak. Common approaches include analyzing pressure sensitivity to leaks 
\cite{Perez2014} 
or solving the inverse hydraulic problem, applying optimization techniques to calibrate network parameters or nodal consumption from pressure and flow data \cite{Sophocleous2019}. 
Model-based methods can achieve high accuracy if the hydraulic model is well-calibrated, but this is a challenging task given the complexity of networks, 
mathematical models 
and potential modeling errors. 

Advancements in data analysis and machine learning helped to reduce the reliance on well-calibrated hydraulic models while achieving comparable node-level accuracy\footnote{In the context of leak localization, accuracy can be classified as ``node-level'', which implies identifying the exact network element with the leak, or ``area-level'', where a broader search area is indicated.} to model-based methods, leading to the development of mixed model-based/data-driven methods. They typically use the model to generate training samples for the potential leaks at different network conditions, to then learn from these data. Common strategies include the use of artificial neural networks (ANN) \cite{Caputo2003}, support vector machine (SVM) \cite{Zhang2016} and deep learning \cite{Zhou2019}. 

Recently, the interest in removing the necessity of a hydraulic model has driven the advancement in data-driven methods, which avoid model-related problems at the cost of a high dependence on sensor readings. Most data-driven approaches exploit the structure of the network, integrating graph-related properties \cite{Alves2021} or methods such as interpolation \cite{Soldevila2020,Boatwright2023}. Their promising performance inspired us to explore interpolation-based leak localization, producing methods such as Graph-based State Interpolation (GSI) \cite{RomeroBen2022}, which only requires hydraulic head readings from distributed sensors and structural network information (pipe length and connectivity). This method was further extended to derive the interpolation weights from an approximated head-to-head relation based on the Hazen-Williams equation, leading to Analytical Weighting GSI (AW-GSI) \cite{Irofti2023}. 

Most of the mentioned methodologies, developed in previous works, rely solely on pressure sensors due to their lower cost, easier installation, potential use in network control 
and the higher leak sensitivity of pressure with respect to flow \cite{Santos2022}. However, other sensor types might be present in the WDN, such as flow sensors (typically installed in water inlets) or demand sensors, normally denoted as Automated Meter Reading (AMR) \cite{Ali2022}. Therefore, a natural step to improve state interpolation involves integrating all the available measurements in the WDN. To address the limitations of GSI and AW-GSI, which cannot handle the non-linear relation between heads and flow/demand, we recently proposed an upgrade of these methods that fuses pressure and demand information through an Unscented Kalman Filter (UKF) \cite{Julier1997}, referred to as UKF-(AW)GSI \cite{RomeroBen2024b}. 

\noindent\textbf{Contributions}: UKF-(AW)GSI models the network state using nodal hydraulic heads, given their importance during leak localization. However, other hydraulic variables such as flow and nodal demand are also of interest for water utilities, not only for leak management but also for network operational control. 
While flows and demands can be computed from hydraulic heads using the Hazen-Williams and mass conservation equations, the higher number of pipes relative to nodes means that a certain head vector can be obtained from multiple flow solutions. As a result, flow and demand estimations obtained from the head states may not accurately represent the actual conditions within the network. To address this issue, we propose an extension of UKF-(AW)GSI that integrates flow and head estimation together. This method, denoted as Dual UKF-(AW)GSI or D-UKF-(AW)GSI, runs two separate estimation processes: one considers heads as the state, while the other treats flows as the state. These processes are linked through virtual measurements, so that the state of one estimator is used to feed the measurement vector of the other. This allows both processes to yield a consistent dual solution for heads and flows. The performance of our method is tested over the Modena \cite{Bragalli2012} and L-TOWN benchmarks \cite{Vrachimis2022}. 

\noindent\textbf{Notation}: We denote scalars using plain lowercase letters, whereas vectors and matrices are represented in bold, using lowercase and uppercase letters respectively. Considering an arbitrary matrix $\bm{A}$, its $i$-th column is denoted as $\bm{a}_i$, whereas $a_{ij}$ indicates the $i$-$j$ element.  Nevertheless, sub-indices $h$ and $q$ indicate that the corresponding variable belongs either to the head or flow estimation process respectively. Within estimation processes, the use of a hat indicates approximation, such as $\bm{\hat{A}}$ being the approximation of $\bm{A}$. Other important notations include $\left(\bm{x}\right)^{y}$, which means that each component of vector $\bm{x}$ is raised to the power of $y$, and $|\cdot|$, which denotes cardinality if applied to sets, and absolute value otherwise.

\noindent\textbf{Outline}: The rest of the paper is organized as follows. Section \ref{section:preliminaries} introduces preliminary key concepts regarding the proposed estimation method. Section \ref{section:methodology} delves into the contributions of this article, starting with a step-by-step analysis of the applicability of Kalman-based filters to the problem at hand, to then explain the details of the proposed sensor fusion methodology, which improves both head and flow estimation. Section \ref{section:case_study} presents the utilized case study, which is an open-source benchmark known as L-TOWN. The obtained results are detailed in Section \ref{section:results}, exploring the performance of the proposed methodology in terms of both state estimation and leak localization capabilities. Finally, Section \ref{section:conclusions} presents the conclusions and outlines several future worklines.

\section{Preliminaries}\label{section:preliminaries}

\subsection{Hydraulic state estimation}\label{subsection:hydraulic_state_estimation}

Within the context of leak localization, the estimation of the hydraulic state of a WDN consists in using hydraulic information available from existing network-wide sensors to retrieve the state of unknown elements. To characterize a WDN, let us consider its modeling through a simple, connected and directed graph $\mathcal{G}=(\mathcal{V},\mathcal{E})$, where $\mathcal{V}$ denotes the set of nodes, i.e., reservoirs and junctions, and $\mathcal{E}$ stands for the set of edges, i.e., pipes.  The sizes of these sets are denoted as $n=\lvert \mathcal{V}\rvert$ and $m = \lvert \mathcal{E}\rvert$. The $i$-$th$ node of the WDN is expressed as $\mathscr{v}_i\in\mathcal{V}$, whereas the $k$-$th$ edge is denoted as $\mathscr{e}_k = (\mathscr{v}_i,\mathscr{v}_j)\in\mathcal{E}$, with $\mathscr{v}_i$ as its source and $\mathscr{v}_j$ as its sink.

\subsubsection{Graph-based State Interpolation (GSI)}\label{subsubsection:GSI}

A well-known solution for addressing data-driven state estimation is Graph-based State Interpolation (GSI) \cite{RomeroBen2022}.
This method uses the hydraulic head (pressure + elevation) at the network nodes as a proxy for the network hydraulic state,
which is recovered from distributed pressure sensor readings from certain network nodes and the network structure, represented by nodal connectivity and pipes lengths. 
In WDNs, the relation between hydraulic heads of adjacent nodes and the flow between them is defined by a non-linear expression such as the Hazen-Williams formula \cite{Brater1996}:
$b_{ki}(h_i-h_j) = \tau_{k}q_k^{1.852}$,
where $h_i$ and $h_j$ are respectively the $i$-$th$ and $j$-$th$ entries of the hydraulic state vector $\bm{h}\in \mathbb{R}^{n}$ and $b_{ki}$ is the element of the incidence matrix $\bm{B}\in \mathbb{R}^{m\times n}$ denoting that node $\mathscr{v}_i$ is the source ($b_{ki}=1$) or the sink ($b_{ki}=-1$) in pipe $\mathscr{e}_k=(\mathscr{v}_i,\mathscr{v}_j)$. Consequently, we have that $b_{kj}= -b_{ki}$, with the rest of elements of the $k$-th row of $\bm{B}$ being zero. 
The resistance coefficient
$\tau_{k}=(10.67 \rho_{k})/(\mu_{k}^{1.852} \delta_{k}^{4.87})$
corresponds to $\mathscr{e}_k$
such that $\rho_{k}, \delta_{k}$ and $\mu_{k}$ are the pipe length, diameter and roughness respectively.
The variable $q_k$ is the flow traversing $\mathscr{e}_k$,
defined here as always non-negative,
where $1.852$ is the Hazen-Williams flow exponent.

The core idea of GSI is the relaxation of the Hazen-Williams equation
by substituting it with the weighted linear expression
$\hat{h}_i = \bm{w}^{\top}_i\bm{\hat{h}} / d_i$,
\noindent where $\bm{\hat{h}}\in \mathbb{R}^{n}$ is the approximated state vector of hydraulic heads. The row vector $\bm{w}^{\top}_i$ is the $i$-$th$ row of the weighted adjacency matrix $\bm{W}\in \mathbb{R}^{n \times n}$, encoding the relation between adjacent nodes, defined here as $w_{ij} = \frac{1}{\rho_{k}}$. Additionally, $d_i$ is the $i$-$th$ diagonal entry of the degree matrix $\bm{D}\in \mathbb{R}^{n \times n}$, where $d_i=\sum_{j=1}^{n} w_{ij}$. 
GSI aims to obtain $\bm{\hat{h}}$ 
by solving the following optimization problem \cite{RomeroBen2022}
\begin{align}\label{eq:GSI_opt}
\min_{\bm{\hat{h}}, \gamma} \quad & \frac{1}{2}\big[\bm{\hat{h}}^T\bm{L}\bm{D}^{-2}\bm{L}\bm{\hat{h}}+\zeta \gamma ^2\big],\\
\label{eq:GSI_opt_b}\textrm{s.t.} \quad & \bm{B} \bm{\hat{h}}\leq \gamma \cdot\bm{1}_{n},\ \gamma > 0,\ \bm{S}\bm{\hat{h}}=\bm{\hat{h}}^{s},  
\end{align}
\noindent where $\bm L=\bm D - \bm W$ is the Laplacian of $\mathcal G$, $\bm S \in \mathbb{R}^{n_{s}\times n}$ is the sensorization matrix ($n_s$ is the number of pressure sensors) that extracts the state values at the sensors, $\bm{\hat{h}}^{s}$ is the vector of head measurements and $\zeta$ is a constant weighting the relative importance between the two sub-objectives. The first sub-objective pursues the minimization of $\sum_{i=1}^{n}\Big[\hat{h}_i - \bm{w}^{\top}_i\bm{\hat{h}} / d_i\Big]^2$, whereas the second is connected to the constraints through the positive slack variable $\gamma$, in order to ensure a well-represented directionality of water flowing through the network.

\subsubsection{Analytical Weighting Graph-Based State Interpolation (AW-GSI)}

Starting from the GSI formulation
where the head measurements are
based on the resulting least-squares weights from \eqref{eq:GSI_opt},
AW-GSI~\cite{Irofti2023} obtains a set of analytical weights
that better approximate the heads in the network
through a local linearization of the Hazen-Williams equations.

Indeed,
given a network node $h_i$
there exists a local neighborhood $\J$,
where the neighbor heads are grouped in
$\bm{h}_{\J}=\bbm \dots & h_j & \dots\ebm^\top\in \mathbb R^{|\J|}$,
such that there is a direct dependence between it and its neighbors head values $h_i = g(\bm{h}_\J)$.
The central results in~\cite{Irofti2023}
show that we can write this as
$h_i=\mathbf g_i(\bm{h}_{\J})\approx \bar{h}_i+\sum\limits_{j:\: b_{ij}\neq 0}\eta_{ij}\cdot (h_j-\bar{h}_j),$
where the weights $\eta_{ij}$ are given by the gradient $\nabla g$ as follows: 
$\eta_{ij}=(\tau_{k}^{-0.54} \left|\bar{h}_i-\bar{h}_j\right|^{-0.46})/(\sum\limits_{u:\:  i \sim u} \tau_{k}^{-0.54}\left|\bar{h}_i-\bar{h}_u\right|^{-0.46})$,
\noindent where we denoted with $\bm{{\bar{h}}}$ a state vector following
Hazen-Williams,
and $i \sim u$ indicates that $\mathscr{v}_i$ and $\mathscr{v}_u$ are adjacent. 
The results~\cite{Irofti2023} showed significant improvements when using these analytical weights compared to the standard GSI weight.
In current work, we propose to improve these further through sensor fusion techniques
coupled with the proposed dual UKF architecture.

\subsection{Sensor fusion}\label{subsection:sensor_fusion}

The effectiveness of the methods described in Section \ref{subsection:hydraulic_state_estimation} depends not only on the number of sensors installed but also on their ability to integrate different types of measurements. 
Kalman Filter (KF) based methods are among the most widely used strategies for sensor fusion, due to their proven efficacy in solving a broad range of problems across several disciplines. 

\subsubsection{Kalman Filter}\label{subsubsection:KF}

The Kalman filter \cite{Kalman1960} is a well-known algorithm in control theory for state estimation in linear dynamical systems given in state-space form, which operates by recursively incorporating noisy measurements over time. The optimality of these estimates holds under specific assumptions, primarily the linearity of the system and the Gaussian distribution of the noise. Consider a linear system as follows: 
$\bm{x}^{[k]} = \bm{F}\bm{x}^{[k-1]} + \bm{B}^u\bm{u}^{[k-1]} + \bm{w}^{[k-1]}$,
$\bm{y}^{[k]} = \bm{G}\bm{x}^{[k]} + \bm{v}^{[k]}$,
where $\bm{x}^{[k-1]}$ is the state vector (at the $k$-th iteration), $\bm{F}$ is the state transition matrix, $\bm{B}^u$ is the control-input matrix, $\bm{u}^{[k-1]}$ is the control-input vector and $\bm{w}^{[k-1]} \sim \mathcal{N}(\bm{0},\bm{Q})$ is the process noise vector, following a zero-mean Gaussian distribution with covariance matrix $\bm{Q}$. Moreover, $\bm{y}^{[k]}$ is the measurement vector, $\bm{G}$ is the measurement matrix and $\bm{v}^{[k]}\sim \mathcal{N}(\bm{0},\bm{R})$ is the measurement noise vector following a Gaussian distribution with zero mean and covariance matrix $\bm{R}$. 

The Kalman filter algorithm includes the \textit{prediction} step, modelling the dynamical evolution of the system state, and the \textit{measurement update} step, which corrects the state via measurements. The \textit{prediction} step is defined through the equations:
$\bm{\hat{x}}^{[k]}_{-} = \bm{F}\bm{\hat{x}}^{[k-1]} + \bm{B}^u\bm{u}^{[k-1]}$,
$\bm{\hat{P}}^{[k]}_{-} = \bm{F}\bm{\hat{P}}^{[k-1]}\bm{F}^{\top} + \bm{Q}$,
where $\bm{\hat{x}}^{[k-1]} = \mathbb{E}\left[\bm{x}^{[k-1]}\right]$ is the estimated state and $\bm{\hat{P}}^{[k-1]} = \mathbb{E}\left[\left(\bm{x}^{[k-1]} - \bm{\hat{x}}^{[k-1]}\right)\left(\bm{x}^{[k-1]} - \bm{\hat{x}}^{[k-1]}\right)^{\top}\right]$ is the state error covariance matrix, which quantifies the estimation uncertainty. The \textit{measurement update} step is defined by:
$\bm{K}^{[k]} = \bm{\hat{P}}^{[k]}_{-}\bm{G}^{\top}\left(\bm{G}\bm{\hat{P}}^{[k]}_{-}\bm{G}^{\top} + \bm{R}\right)^{-1}$,
$\bm{\hat{x}}^{[k]} = \bm{\hat{x}}^{[k]}_{-} + \bm{K}^{[k]}\left(\bm{y}^{[k]} - \bm{G}\bm{\hat{x}}^{[k]}_{-}\right)$,
$\bm{\hat{P}}^{[k]} = \left(\bm{I}_n - \bm{K}^{[k]}\bm{G}\right)\bm{\hat{P}}^{[k]}_{-}$,
where $\bm{K}^{[k]}$ is the Kalman gain, and $\bm{y}^{[k]}$ is the actual measurement vector, containing the sensor readings.

\subsubsection{Non-linear Kalman Filter}

Many processes cannot be represented through linear models. This is the case for WDNs, where, as shown in Section~\ref{subsubsection:GSI},
the relationship between flow and head is non-linear. Thus, exploring extensions of the Kalman Filter to non-linear systems is of particular interest.

\paragraph{Extended Kalman Filter}

A widely employed alternative method is the Extended Kalman Filter (EKF) \cite{Kalman1961}, which deals with non-linear systems by linearizing the process and measurement models at the current estimated state.
A non-linear system can be given by 
$\bm{x}^{[k]} = \mbox{\textbf{f}}(\bm{x}^{[k-1]},\bm{u}^{[k-1]}) + \bm{w}^{[k-1]}$
and
$\bm{y}^{[k]} = \mbox{\textbf{g}}(\bm{x}^{[k]}) + \bm{v}^{[k]}$,
\noindent where $\mbox{\textbf{f}}$ and $\mbox{\textbf{g}}$ are respectively the non-linear process and measurement functions. 
The EKF computes the Jacobian matrices of $\mbox{\textbf{f}}$ and $\mbox{\textbf{g}}$ at each time step
$\bm{F}^{[k]} = \left.\partial \mbox{\textbf{f}} / \partial \bm{x} \right|_{\bm{x}^{[k]},\bm{u}^{[k]}}$, $\bm{G}^{[k]} = \left.\partial \mbox{\textbf{g}}/\partial \bm{x}\right|_{\bm{x}^{[k]}}$.

EKF is analogue to KF when the dynamic state transition and measurement matrices, $\bm{F}^{[k]}$ and $\bm{G}^{[k]}$, replace their static counterparts, $\bm{F}$ and $\bm{G}$
from Section~\ref{subsection:sensor_fusion}.
The effectiveness of EKF has been demonstrated over the years,
although estimation errors and potential convergence issues may arise from considering only first-order terms in the Jacobian linearization.

\paragraph{Unscented Kalman Filter}\label{paragraph:UKF_preliminaries}

Recently, the Unscented Kalman Filter (UKF) \cite{Julier1997} emerged as an alternative to EKF, mitigating its limitations without adding extra computational cost. UKF utilizes the actual non-linear models, approximating the distribution of the random variable representing the state. Unlike EKF, which propagates only the state mean and covariance, UKF propagates a set of "sigma" points from the original distribution through the non-linear model.
Each UKF iteration is composed of three stages, namely \textit{prediction}, \textit{measurement propagation} and \textit{correction}. 

\noindent The \textit{prediction} phase is composed of the following steps:
$\bm{\mathcal{X}}^{[k-1]} = \begin{bmatrix}
     \bm{\hat{x}}^{[k-1]} & \bm{\hat{x}}^{[k-1]} \pm \eta\sqrt{\bm{\hat{P}}^{[k-1]}}
\end{bmatrix}$,
$\bm{\mathcal{X}}^{[k]}_{-} = \mbox{\textbf{f}}(\bm{\mathcal{X}}^{[k-1]},\bm{u}^{[k-1]})$,
$\bm{\hat{x}}^{[k]}_{-} = \sum_{i=0}^{2n} w_i^{(m)}\bm{\mathcal{X}}^{[k]}_{i,-}$,
$\bm{\hat{P}}^{[k]}_{-} = \sum_{i=0}^{2n} w_i^{(c)} (\bm{\mathcal{X}}^{[k]}_{i,-} - \bm{\hat{x}}^{[k]}_{-})(\bm{\mathcal{X}}^{[k]}_{i,-} - \bm{\hat{x}}^{[k]}_{-})^{\top} + \bm{Q}$,
where $\eta = \sqrt{n+\lambda}$, and $\lambda = n(\alpha^2-1)$, are scaling parameters ($\alpha$ sets the spread of the "sigma" points around the mean). The state mean and covariance are reconstructed through the scaled unscented transform (SUT), using a weighting approach over the prior "sigma" points, with weights $w_0^{(m)} = \frac{\lambda}{n+\lambda}, w_0^{(c)} = \frac{\lambda}{n+\lambda} + (1 - \alpha^2 + \beta^2)$ and $w_i^{(m)} = w_i^{(c)} = \frac{1}{2(n+\lambda)},\: \forall i=1,2,\ldots,2n$, with $(m)$ and $(c)$ referring to the mean and  covariance, respectively ($\beta$ enables the integration of prior knowledge about the actual distribution of the state).   

The \textit{measurement propagation} step is defined by
$\bm{\mathcal{X}}^{[k]}_- = \begin{bmatrix}
    \bm{\hat{x}}^{[k]}_- &  \bm{\hat{x}}^{[k]}_- \pm \eta\sqrt{\bm{\hat{P}}^{[k]}}_-
\end{bmatrix}$,
$\bm{\mathcal{Y}}^{[k]}_{-} = \mbox{\textbf{g}}(\bm{\mathcal{X}}^{[k]}_{-})$,
$\bm{\hat{y}}^{[k]}_{-} = \sum_{i=0}^{2n} w_i^{(m)}\bm{\mathcal{Y}}^{[k]}_{i,-}$,
where the SUT is again used, retrieve the prior measurements vector from the prior state vector
to computing the measurement $\bm{\hat{P}}^{[k]}_{yy}$ and cross covariances $\bm{\hat{P}}^{[k]}_{xy}$
such that
$\bm{\hat{P}}^{[k]}_{yy} = \sum_{i=0}^{2n} w_i^{(c)}(\bm{\mathcal{Y}}^{[k]}_{i,-} - \bm{\hat{y}}^{[k]}_{-})(\bm{\mathcal{Y}}^{[k]}_{i,-} - \bm{\hat{y}}^{[k]}_{-})^{\top} + \bm{R}$,
$\bm{\hat{P}}^{[k]}_{xy} = \sum_{i=0}^{2n} w_i^{(c)}(\bm{\mathcal{X}}^{[k]}_{i,-} - \bm{\hat{x}}^{[k]}_{-})(\bm{\mathcal{Y}}^{[k]}_{i,-} - \bm{\hat{y}}^{[k]}_{-})^{\top}$.


\noindent Finally, the \textit{correction} step is performed equivalently to the \textit{measurement update} stage of the linear KF:
$\bm{K}^{[k]} = \bm{\hat{P}}^{[k]}_{xy}(\bm{\hat{P}}^{[k]}_{yy})^{-1}$,
$\bm{\hat{x}}^{[k]} = \bm{\hat{x}}^{[k]}_{-} + \bm{K}^{[k]}(\bm{y}^{[k]} - \bm{\hat{y}}^{[k]}_{-})$,
$\bm{\hat{P}}^{[k]} = \bm{\hat{P}}^{[k]}_{-} - \bm{K}^{[k]}\bm{\hat{P}}^{[k]}_{yy}(\bm{K}^{[k]})^{\top}$.

\begin{remark}\label{remark:UKF_preliminaries}
    Unlike the EKF, which relies only on first-order terms (Taylor series expansion), the UKF provides accurate third order approximations for Gaussian distributions and at least the second order for non-Gaussian ones \cite{VanDerMerwe2004}. The precision of higher-order moments in the UKF depends on the choice of the scaling parameters $\alpha$ and $\beta$. \eor
\end{remark}

\section{Methodology}\label{section:methodology}

Sensor fusion methods, such as the KF strategies described in Section \ref{subsection:sensor_fusion}, are a solution for integrating measurements in WDNs. They can be coupled with interpolation strategies like GSI or AW-GSI to improve estimation accuracy: the state vector reconstructed by any of these interpolation methods is used as the initial guess for the KF algorithm, which iterates between the state prediction and data assimilation steps until the stop condition is met.

\subsection{Improving state estimation through KF}\label{subsection:improving_state_KF}

An initial integration of GSI/AW-GSI with Kalman Filter based methods can be attained by considering the linear case. In this scenario, pressure measurements continue to be the only source of information for the network hydraulic state. 

\subsubsection{Prediction step}

The state prediction is defined as in Section~\ref{subsection:sensor_fusion}.
In this case, the estimated state is represented by the estimated hydraulic head vector $\bm{\hat{h}}^{[k]}$ (at the $k$-th step), with no input $\bm{u}^{[k]}$, and where $\bm{F}$ is defined as follows: 

\begin{equation}\label{eq:KF(AW)GSI_prediction}
    \bm{\hat{h}}^{[k]}_- = \bm{F}\bm{\hat{h}}^{[k-1]} = \left(\epsilon \bm{I}_n + (1-\epsilon)\bm{\Phi}^{-1}\bm{\Omega})\right)\bm{\hat{h}}^{[k-1]},
\end{equation}

\noindent where $\bm{\Omega}$ and $\bm{\Phi}$ are a weighted adjacency matrix and its degree matrix (derived from GSI/AW-GSI), and $0\leq\epsilon\leq1$ is a weight measuring the relevance of the identity matrix, which aims to keep the previous state after prediction, and the diffusion matrix $\bm{\Psi}=\bm{\Phi}^{-1}\bm{\Omega}$, which seeks to diffuse the previous state considering the relationship among neighboring nodes. 

In KF based methods, the estimation from this linear process model (prior) is provided to a measurement correction process to retrieve the actual (posterior) estimation, in order to avoid a potential degradation in performance. The relevance of the measurement correction step 
can be shown by analyzing the evolution of the state through the iterations if a process model such as \eqref{eq:KF(AW)GSI_prediction} is used without the measurement update phase. In this scenario, the state estimation at $[k]$ can be easily posed with respect to the initial estimation as $\bm{\hat{h}}^{[k]} = \bm{F}^k\bm{\hat{h}}^{[0]}$.

\begin{prop}\label{proposition:1}
    The $k$-$th$ power of $\bm{F}$ in \eqref{eq:KF(AW)GSI_prediction}, i.e., $\bm{F}^k$, is power convergent with
    $\lim_{k\to\infty} \bm{F}^k = \bm{\tilde{F}}$,
    \noindent where $\tilde{f}_{ij} > 0, \; \forall i,j = 1,2,..., n$. Moreover
    $\bm{\tilde{F}}\bm{x} = \delta\bm{v}_1$,
    \noindent where $\bm{x}$ is any random vector, $\delta$ is a scalar dependent on $\bm{x}$, and $\bm{v}_1$ is the eigenvector associated to the largest eigenvalue of $\bm{F}$, i.e., $\lambda_1 = 1$, with all the elements in $\bm{v}_1$ being equal.
\end{prop}
\begin{proof}
    The effect of applying a linear process model, defined by a matrix $\bm{F}$, to a state vector depends on the eigenvalues of this matrix. The Perron-Frobenius theorem \cite{Pillai2005} states that if $\bm{F}$ is a stochastic\footnote{Note that, in this context, "stochastic" refers to square matrices with non-negative elements whose rows sum up to 1.}, irreducible and primitive matrix, then $\lambda_1 = 1$ is an eigenvalue of $\bm{F}$ with algebraic (and geometric) multiplicity 1, associated to a positive eigenvector $\bm{v}_1 > 0$, and the rest of  eigenvalues of $\bm{F}$ satisfy that $|\lambda_i|< 1,\; \forall i=2,...,n$. 
    
    First, for $\bm{F}$ to be stochastic, the conditions $f_{ij}\geq 0, \; \forall i,j=1,2,...,n$ and $\sum_{j=1}^{n} f_{ij} = 1$ must hold. About the first condition, note that $0\leq\epsilon\leq 1$ and $\psi_{ij}\geq 0, \; \forall i,j=1,2,...,n$ if $\omega_{ij}\geq 0, \; \forall i,j=1,2,...,n$, considering that $\phi_{ii} = \sum_{j=1}^{n} \omega_{ij}$ (and $\phi_{ij}=0$ if $i\neq j$). Therefore, $f_{ij}\geq 0, \; \forall i,j=1,2,...,n$ is fulfilled. Moreover, about the second condition, note that it is fulfilled by $\bm{\Psi}$ because it is a diffusion matrix, i.e., $\sum_{j=1}^{n} \psi_{ij} = \sum_{j=1}^{n} \frac{1}{\phi_{ii}}\omega_{ij} = \frac{1}{b\phi_{ii}} \sum_{j=1}^{n} \omega_{ij} = \frac{1}{\phi_{ii}}\phi_{ii} = 1$. Due to the definition of $\bm{F}$ in \eqref{eq:KF(AW)GSI_prediction} and the range of $\epsilon$, we have that $\sum_{j=1}^{n} f_{ij} = \epsilon + (1-\epsilon)\sum_{j=1}^{n} \psi_{ij} = 1$. 

    Then, for $\bm{F}$ to be irreducible, the directed graph of $\bm{F}$, denoted as $\mathcal{G}^F = (\mathcal{V}^F,\mathcal{E}^F)$, must be strongly connected. A graph is strongly connected if there exists a directed path between any two distinct vertices. By definition, $\bm{F}$ has positive non-zero values in both $f_{ij}$ and $f_{ji}$ if $\mathscr{v}_i$ and $\mathscr{v}_j$ are connected, because this also holds for $\bm{\Psi}$ (additionally, the identity matrix in \eqref{eq:KF(AW)GSI_prediction} can be considered to represent self-loops in $\mathcal{G}^F$, although they are not relevant to check the irreducibility of $\bm{F}$). This translates into the existence of a directed edge in $\mathcal{G}^F$ from $\mathscr{v}_i^F$ to $\mathscr{v}_j^F$, and another directed edge from $\mathscr{v}_j^F$ to $\mathscr{v}_i^F$. As $\bm{\Psi}$ is the weighted adjacency matrix of a connected graph $\mathcal{G}$, a directed path between any two distinct vertices in $\mathcal{G}^F$ must exist, and thus $\mathcal{G}^F$ must be strongly connected and $\bm{F}$ is irreducible.

    Additionally, for $\bm{F}$ to be primitive, there must exist a positive integer $k^+$ such that all entries of $\bm{F}^{k^+}$ are positive. $\bm{F}$ is a stochastic matrix, and therefore $f_{ij}$ can be regarded as the probability of transitioning from $\mathscr{v}_i$ to $\mathscr{v}_j$, with the transition being defined as a multiplication by $\bm{F}$. $\mathcal{G}$ is connected, so there must exist a path between each possible pair of nodes. Thus, with sufficient transitions, it must be possible to arrive from any node to any other node. If the longest possible path in the graph requires $k^+$ transitions (steps), then $\bm{F}^{k^+}>0$.

    Considering that $\bm{F}$ is shown to verify the requirements of the Perron-Frobenius theorem, the conditions about the eigenvalues of $\bm{F}$ by this theorem apply. Regarding that $\bm{F} = \bm{V}\bm{\Lambda}\bm{V}^{-1}$, it is well-known that $\bm{F}^k = \bm{V}\bm{\Lambda}^k\bm{V}^{-1}$.
    If $\bm{\Lambda}= \mbox{diag}(\lambda_1,\lambda_2,\ldots,\lambda_n)$, then $\bm{\Lambda}^k = \mbox{diag}(\lambda_1^k,\lambda_2^k,\ldots,\lambda_n^k)$. As the Perron-Frobenius theorem guarantees that $\lambda_1 = 1$ and $|\lambda_i|<1, \; \forall i=2,\ldots,n$, when $k\rightarrow \infty$, we have that $\lambda_1^k = 1$ and $\lambda_i^k \rightarrow 0, \; \forall i=2,\ldots,n$, leading to $\lim_{k\to\infty} \bm{\Lambda}^k = \mbox{diag}(1,0,\ldots,0)$,  and therefore the $k$-$th$ power of $\bm{F}$ can be defined as $\lim_{k\to\infty} \bm{F}^k = \lim_{k\to\infty} \bm{V}\bm{\Lambda}^k\bm{V}^{-1}= \bm{V}\mbox{diag}(1,0,\ldots,0)\bm{V}^{-1} =  \bm{v}_1\bm{v}^{inv}_{(1,:)}$  
    %
    %
    %
    where $\bm{v}^{inv}_{(1,:)}$ is the first row of $\bm{V}^{-1}$. Considering that $\bm{F}$ is a stochastic matrix, $\bm{F}\bm{1}_n = \bm{1}_n$. This implies that $\lambda_1 = 1$ is an eigenvalue associated to an eigenvector $v_1$ that is proportional to $\bm{1}_n$, and because $\bm{F}$ fulfills the Perron-Frobenius theorem, it is guaranteed that $\lambda_1$ is the largest eigenvalue. Finally, considering that approximation $\bm{F}^k\bm{x} \approx \bm{v}_1\bm{v}^{inv}_{(1,:)}\bm{x} = \delta\bm{v}_1$
    %
    %
    where $\delta = \bm{v}^{inv}_{(1,:)}\bm{x}$ asymptotically becomes an equality when $k\rightarrow \infty$, concludes the proof.
\end{proof}

Therefore, the steady-state behaviour of the linear process model leads to a constant vector, which would not capture the differences in hydraulic head among the network junctions. 

\subsubsection{Measurement update step}

The measurement update assimilates sensor readings. In the linear case, only pressure measurements can be considered, because the relationship between heads and flows or demands is non-linear. Thus, the measurement function can be posed simply as
\begin{equation}\label{eq:KF(AW)GSI_measurement}
    \bm{y}^{[k]}  = \bm{S}\bm{h}^{[k]},
\end{equation}
where the sensorization matrix $\bm{S}$ associates the measured heads with the node states corresponding to the installed sensors. 
An analysis about the evolution of the state estimation can be again performed in order to remark the importance of adding the measurement update step. The data assimilation steps of the Kalman Filter, posed using the system defined by \eqref{eq:KF(AW)GSI_prediction} and \eqref{eq:KF(AW)GSI_measurement}, can be expressed as
$\bm{\hat{h}}^{[k]} = \bm{\hat{h}}^{[k]}_{-} + \bm{K}_h^{[k]}(\bm{h}_s - \bm{S}\bm{\hat{h}}^{[k]}_{-}) = \bm{F}\bm{\hat{h}}^{[k-1]} + \bm{K}_h^{[k]}(\bm{h}_s - \bm{S}\bm{F}\bm{\hat{h}}^{[k-1]})$
thus
\begin{equation}\label{eq:KF(AW)GSI_pred+meas}
    \bm{\hat{h}}^{[k]} = 
    \left(\bm{I}_n - \bm{K}_h^{[k]}\bm{S}\right)\bm{F}\bm{\hat{h}}^{[k-1]} + \bm{K}_h^{[k]}\bm{h}_s, 
\end{equation}
\noindent where $\bm{K}_h^{[k]}$ is the Kalman gain and $\bm{h}_s$ is the vector of actual head measurements, which is constant during the estimator operation because our method operates over individual samples. Since measurements are extracted from the actual state $\bm{h}$ through $\bm{S}$, \eqref{eq:KF(AW)GSI_pred+meas} can be reformulated as
$\bm{\hat{h}}^{[k]} = \left(\bm{I}_n - \bm{K}_h^{[k]}\bm{S}\right)\bm{F} \bm{\hat{h}}^{[k-1]} + \bm{K}_h^{[k]}\bm{S}\bm{h}$.
If the estimation error at the next iteration is defined as
$\bm{e}^{[k]} = \bm{h} - \bm{\hat{h}}^{[k]} = \bm{h} - (\bm{I}_n - \bm{K}_h^{[k]}\bm{S})\bm{F} \bm{\hat{h}}^{[k-1]} - \bm{K}_h^{[k]}\bm{S}\bm{h} = 
    (\bm{I}_n - \bm{K}_h^{[k]}\bm{S})\bm{h} - (\bm{I}_n - \bm{K}_h^{[k]}\bm{S})\bm{F} \bm{\hat{h}}^{[k-1]} = 
    (\bm{I}_n - \bm{K}_h^{[k]}\bm{S})(\bm{h} - \bm{F} \bm{\hat{h}}^{[k-1]}) =
    (\bm{I}_n - \bm{K}_h^{[k]}\bm{S})(\bm{h} - \bm{\hat{h}}^{[k]}_-)$,
then
\begin{equation}\label{eq:KF(AW)GSI_pred+meas_3}
    \bm{e}^{[k]} = \left(\bm{I}_n - \bm{K}_h^{[k]}\bm{S}\right)\bm{e}^{[k]}_-.
\end{equation}
The expression in \eqref{eq:KF(AW)GSI_pred+meas_3} defines the evolution of the estimation error after the prediction step to the estimation error after the measurement update step. An analysis of the system matrix can clarify the effect of the data assimilation step.
\begin{prop}\label{proposition:2}
    With the notation from \eqref{eq:KF(AW)GSI_pred+meas_3}, the following equivalence holds
    $\bm{I}_n - \bm{K}_h^{[k]}\bm{S} = \left(\bm{I}_n + \bm{P}^{[k]}_{h,-}\bm{S}^{\top}\bm{R}_h^{-1}\bm{S}\right)^{-1}$,
    \noindent where $\bm{R}_h$ is the measurement noise covariance matrix for the the data assimilation in \eqref{eq:KF(AW)GSI_measurement}.
\end{prop}
\begin{proof}
    The gain of the linear Kalman Filter, considering the measurement function in \eqref{eq:KF(AW)GSI_measurement}, may be written as:
    $\bm{K}_h^{[k]} = \bm{\hat{P}}^{[k]}_{h,-}\bm{S}^{\top}\left(\bm{S}\bm{\hat{P}}^{[k]}_{h,-}\bm{S}^{\top} + \bm{R}_h\right)^{-1}$.
    Manipulating in both sides yields:
    $\bm{\hat{P}}^{[k]}_{h,-} - \bm{K}_h^{[k]}\bm{S}\bm{\hat{P}}^{[k]}_{h,-} = \bm{\hat{P}}^{[k]}_{h,-} - \bm{\hat{P}}^{[k]}_{h,-}\bm{S}^{\top}(\bm{S}\bm{\hat{P}}^{[k]}_{h,-}\bm{S}^{\top} + \bm{R}_h)^{-1}\bm{S}\bm{\hat{P}}^{[k]}_{h,-}$.
    Recall the Woodbury matrix identity, which states that
    $(\bm{A} + \bm{U}\bm{C}\bm{V})^{-1} = \bm{A}^{-1} - \bm{A}^{-1}\bm{U}(\bm{V}\bm{A}^{-1}\bm{U} + \bm{C}^{-1})^{-1}\bm{V}\bm{A}^{-1}$,
    where $\bm{A}, \bm{U}, \bm{V}$ and $\bm{C}$ are conformable matrices. Adapting the last equality,
    if $\bm{A} = (\bm{\hat{P}}^{[k]}_{h,-})^{-1}$, $\bm{U} = \bm{S}^{\top}$, $\bm{V} = \bm{S}$ and $\bm{C} = \bm{R}_h^{-1}$, it follows that
    $\bm{\hat{P}}^{[k]}_{h,-} - \bm{K}_h^{[k]}\bm{S}\bm{\hat{P}}^{[k]}_{h,-} = [(\bm{\hat{P}}^{[k]}_{h,-})^{-1} + \bm{S}^{\top}\bm{R}_h^{-1}\bm{S}]^{-1}$.

    Extracting $\bm{\hat{P}}^{[k]}_{h,-}$ in the left-hand side, multiplying $\left(\bm{\hat{P}}^{[k]}_{h,-}\right)^{-1}$ in both sides and recalling the properties of the multiplication of two inverse matrices, we get
    $\bm{I}_n - \bm{K}_h^{[k]}\bm{S}
        = ((\bm{\hat{P}}^{[k]}_{h,-})^{-1} + \bm{S}^{\top}\bm{R}_h^{-1}\bm{S})^{-1}(\bm{\hat{P}}^{[k]}_{h,-})^{-1}
        =(\bm{I}_n + \bm{\hat{P}}^{[k]}_{h,-}\bm{S}^{\top}\bm{R}_h^{-1}\bm{S})^{-1}$
    thus concluding the proof.
\end{proof}

Using in \eqref{eq:KF(AW)GSI_pred+meas_3} the equivalence in Proposition \ref{proposition:2} gives:
\begin{equation}\label{eq:KF(AW)GSI_pred+meas_4}
    \bm{e}^{[k]} = \left(\bm{I}_n + \bm{\hat{P}}^{[k]}_{h,-}\bm{S}^{\top}\bm{R}_h^{-1}\bm{S}\right)^{-1}\bm{e}^{[k]}_-.
\end{equation}
Next we explore the update from the posterior error at $[k-1]$, i.e., $\bm{e}^{[k-1]}$, to $\bm{e}^{[k]}_-$ by first recalling their definitions:
$\bm{e}^{[k]}_- = \bm{h} - \bm{\hat{h}}^{[k]}_- = \bm{h} - \bm{F}\bm{\hat{h}}^{[k-1]}$
and
$\bm{e}^{[k-1]} = \bm{h} - \bm{\hat{h}}^{[k-1]}$.
The former comes from the definition of the prediction step, and the later is analogue to the definition in \eqref{eq:KF(AW)GSI_pred+meas_3}. Rewriting the later as $ \bm{\hat{h}}^{[k-1]}= \bm{h} - \bm{e}^{[k-1]}$ and inserting it into the former leads to
    $\bm{e}^{[k]}_- = \bm{h} - \bm{F}\left(\bm{h} - \bm{e}^{[k-1]}\right) = \bm{F}\bm{e}^{[k-1]} + \left(\bm{I}_n - \bm{F}\right)\bm{h}$.
Introducing this 
in \eqref{eq:KF(AW)GSI_pred+meas_4} leads to the update equation for the a posterior error:
\begin{equation}\label{eq:KF(AW)GSI_pred+meas_7}
    \bm{e}^{[k]} =  \bm{M}^{[k]}_s\bm{F}\bm{e}^{[k-1]} + \bm{M}^{[k]}_s\left(\bm{I}_n - \bm{F}\right)\bm{h},
\end{equation}
with shorthand $\bm{M}^{[k]}_s = \left(\bm{I}_n + \bm{\hat{P}}^{[k]}_{h,-}\bm{S}^{\top}\bm{R}_h^{-1}\bm{S}\right)^{-1}$.
\begin{lemma}\label{th:PD-lemma}
Let $P$ be a semi-positive definite matrix and $D$ a diagonal matrix with non-negative entries $d_{ii} \ge 0$,
then matrix $PD$ has non-negative eigenvalues.
\end{lemma}
\begin{proof}
We prove the result by contradiction.
Assume there exists $\lambda < 0$, the eigenvalue of $PD$
and its associated eigenvector $v$. Then
$PD v = \lambda v$
iff
$v^\top D^\top P D v = \lambda (v^\top D^\top v)$
iff
$\underbrace{(Dv)^\top P (Dv)}_{\ge 0} = \lambda v^\top D v 
= \lambda \sum_i \underbrace{d_{ii}}_{\ge 0} \underbrace{v_i^2}_{\ge 0}$
where
we notice that $\lambda$ cannot be negative in the right-hand side.
\end{proof}

\begin{lemma}\label{th:AB-lemma}
The product of two matrices $A$ and $B$, each with their spectrum in the unit circle, has subunitary singular values. 
\end{lemma}
\begin{proof}
Let
$\sigma_{\max}(A)= \max\limits_{\norm{x}=1}\norm{Ax}$.
Assuming\footnote{Since an eigenvector may be multiplied with an arbitrary constant, it follows that we consider normed vectors and that the convexity of coefficients $\alpha_i$ holds.} that $x=\sum_i \alpha_iv_i$ with $\alpha_i\geq 0, \sum_i\alpha_i \leq 1$ and $\norm{v_i}=1$ with $v_i$ the eigenvectors of $A$, then the next inequalities hold
$\norm{Ax} = \norm{A\sum\limits_i \alpha_i v_i}= \norm{\sum\limits_i \alpha_i Av_i}
    =\norm{\sum\limits_i \alpha_i\lambda_i v_i}\leq \sum\limits_i |\alpha_i\lambda_i|\cdot \norm{v_i}
    = \sum\limits_i |\alpha_i\lambda_i|\leq \left(\sum\limits_i \alpha_i\right)\cdot |\lambda_{\max}(A)| \leq |\lambda_{\max}(A)|$.
Recalling that
$\norm{AB} \le \norm{A}\norm{B}$
implies that
$\norm{AB} = \sigma_{\max}(AB) \le \sigma_{\max}(A)\sigma_{\max}(B)$
we have that $\sigma_{\max}(AB) \leq |\lambda_{\max}(A)\lambda_{\max}(B)|<1$ concluding the proof.
\end{proof}

\begin{lemma}\label{th:Av-lemma}
Let $A$ be a matrix with subunitary real eigenvalues,
$-1 \le \lambda(A) \le 1$, then $\norm{Av} \le \norm{v}$ holds for any $v$, a vector of appropriate size.
\end{lemma}
\begin{proof}
Let $\tilde{\bm{v}}=v/\norm{v}$ be the normed vector $v$,
then
$\bm{Av} = \bm{A}\norm{\bm{v}} \tilde{\bm{v}}$,
and from the definition
$\sigma_{max}(A) = \max_{\norm{x}=1}\norm{Ax}$
we get
$\norm{Av} = \norm{v}\norm{A\tilde{v}} \le \norm{v} \sigma_{max}(A)$, 
where in the last inequality we used that
$A$ has subunitary eigenvalues.
\end{proof}

\begin{theorem}
The iteration in \eqref{eq:KF(AW)GSI_pred+meas_7} is monotonically decreasing.
\end{theorem}
\begin{proof}
Note that in \eqref{eq:KF(AW)GSI_pred+meas_4},
$\hat{\bm{P}}^{[k]}_{h,-}$ is the prior state error covariance matrix as defined in Section~\ref{subsection:sensor_fusion}, which is also a semi-positive definite matrix. Further, matrix $\bm R_h$ is the measurement error covariance matrix,
where the noise is modeled as zero-mean Gaussian noise
drawn iid from the normal distribution as discussed in Section~\ref{subsection:sensor_fusion}.
This implies that $\bm{R}_h$ is diagonal, and thus its inverse is also diagonal. Then, considering that
$\bm{S}$ is the sensorization matrix, the matrix product $\bm{S}^T\bm{R}_h^{-1}\bm{S}$ is also diagonal as it picks the rows and columns pairs in $\bm{R}_h$ corresponding to the WDN sensors.

Applying Lemma~\ref{th:PD-lemma} for matrix $\bm{\hat{P}}^{[k]}_{h,-}\bm{S}^{\top}\bm{R}_h^{-1}\bm{S}$ shows that its eigenvalues are real and non-negative.
Thus, per Gershgorin's circle theorem,
adding $\bm{I}_n$ to this quantity
shifts the spectrum by one
guaranteeing that all eigenvalues are supra-unitary.
Finally the inversion in \eqref{eq:KF(AW)GSI_pred+meas_4}
leads to a matrix with positive and real eigenvalues within the unit circle.

Through the norm triangle inequality property and  Lemma~\ref{th:Av-lemma} we have the next chain of inequalities:
$\norm{\bm{e}^{[k]}} =\norm{\bm{M}^{[k]}_s\bm{F}\bm{e}^{[k-1]}+\left(\bm{I}_n - \bm{F}\right)\bm{h}}
    \leq \norm{\bm{M}^{[k]}_s\bm{F}\bm{e}^{[k-1]}}+\norm{\left(\bm{I}_n - \bm{F}\right)\bm{h}}
    \leq \sigma_{\max}(\bm{M}^{[k]}_s\bm{F})\norm{\bm{e}^{[k-1]}}
        +\sigma_{\max}(\bm{M}^{[k]}_s(\bm{I}_n - \bm{F}))\norm{\bm{h}}
    \leq \bar\sigma \cdot \left(\norm{\bm{e}^{[k-1]}}+ \norm{\bm{h}}\right)$,
where $\bar \sigma$ is a shorthand denoting the upper bound for both $\sigma_{\max}\left(\bm{M}^{[k]}_s\bm{F}\right)$ and $\sigma_{\max}\left(\bm{M}^{[k]}_s\left(\bm{I}_n - \bm{F}\right)\right)$, at any $[k]$. Iterating from $[0]$ to $[k]$ leads to 
$\norm{\bm{e}^{[k]}} \leq \bar \sigma^{k-1} \norm{\bm{e}^{[0]}}+(\bar \sigma+\ldots+\bar \sigma^{k})\norm{\bm{h}}
=\bar \sigma^{k-1} \norm{\bm{e}^{[0]}} +
 (\bar \sigma-\bar \sigma^{k+1})/(1-\bar \sigma)\norm{\bm{h}}$.
Applying Lemma~\ref{th:AB-lemma} for matrices
$\bm{M}^{[k]}_s$ and $\bm{F}$,
both with subunitary real eigenvalues,
we obtain that the product matrices $\bm{M}^{[k]}_s\bm{F}$ and $\bm{M}^{[k]}_s\left(\bm{I}_n - \bm{F}\right)$ also have this property. This allows to state that $0<\bar \sigma < 1$, which introduced in
the last inequality
shows that the estimation error is bounded in the initial error and the term $\bm h$ and that, asymptotically (when $k\rightarrow \infty$), it depends only on the second term, thus concluding the proof.
\end{proof}
\begin{remark}
    The previous theoretical framework allows us to confirm the stability of the required Kalman Filter implementation, composed by a linear state prediction function (through a diffusion matrix as the state prediction matrix) and the availability of only head measurements. In the following sections, the state estimation methodology is extended to include demand measurements, which are related to the system states (hydraulic heads) through a non-linear equation. Therefore, the linear theory in Section \ref{subsection:improving_state_KF} help us to motivate the promising performance of Kalman-based algorithms. The necessary extension to address the problem of inclusion of non-linear measurements is handled with the use of the Unscented Kalman Filter. The analysis of the stability of this algorithm is beyond of the scope of this article, although several results in the state of the art analyze and reduce associated numerical problems leading to instability, such as the development of upgraded algorithms like the square-root UKF \cite{VanDerMerwe2004}, and the development of the UKF-related theory \cite{Menegaz2015}.
    \eor
\end{remark}
\subsection{Including demand measurements}\label{subsection:including_demands} 

The non-linear estimation capabilities of the UKF, explained in Section \ref{paragraph:UKF_preliminaries}, enable the integration of demand measurements in the data assimilation function, as demonstrated with UKF-(AW)GSI in \cite{RomeroBen2024b}. 

In this method, the \textit{prediction} step employs a linear function, simplifying the prediction process due to its equivalency with the prediction operation in the KF for a linear prediction function $\mbox{\textbf{f}}$. This equivalence arises from the SUT's ability to perfectly capture linearity without requiring further approximations \cite{Simon2006}: 
\begin{equation}\label{eq:UKF-GSI_pred_steps}
\bm{\hat{h}}^{[k]}_{-} = \bm{F}\bm{\hat{h}}^{[k-1]},\ \ 
\bm{\hat{P}}^{[k]}_{h,-} = \bm{F}\bm{\hat{P}}^{[k-1]}_h\bm{F}^{\top} + \bm{Q}_h 
\end{equation}
\noindent where $\bm{F}$ and $\bm{Q}_h$ are respectively the state transition and process noise covariance matrices for the prediction in \eqref{eq:KF(AW)GSI_prediction}. Note that $\bm{u}^{[k-1]}$ is removed because there is no input in the process. In our case, the weighting parameter $\epsilon$ in the definition of $\bm{F}$ is defined as $\epsilon = \frac{n_c}{n}$, where $n_c$ is the number of AMRs. This decision is justified by the importance of demand information within the estimation process. If demand data is abundant, the relative weight of the data assimilation process should be higher, and thus we should maintain the state from the previous iteration. Nonetheless, if demand data is scarce, the prediction function can make use of the diffusion matrix to extend the corrected state estimates to unobserved regions.

Based on UKF,
a non-linear expression appears in the \textit{measurement propagation} step:

\small 
\begin{subequations}
\label{eq:UKF-GSI_meas_prop_steps}
\begin{align}
\label{eq:UKF-GSI_meas_prop_steps_a}
 \bm{\mathcal{H}}^{[k]}_- &= \begin{bmatrix}
    \bm{\hat{h}}^{[k]}_- & \bm{\hat{h}}^{[k]}_- \pm \eta\sqrt{\bm{\hat{P}}^{[k]}_{h,-}} 
\end{bmatrix}\\
\label{eq:UKF-GSI_meas_prop_steps_b} \bm{\mathcal{Y}}^{[k]}_{-} &= \mbox{\textbf{g}}(\bm{\mathcal{H}}^{[k]}_{-}) = \begin{bmatrix}
    \bm{S}\bm{\mathcal{H}}^{[k]}_{-} \\ -\left(\hat{\bm{B}}^{[k]}_c\right)^{\top}\left(\bm{T}^{-1}\hat{\bm{B}}^{[k]}\bm{\mathcal{H}}^{[k]}_{-}\right)^{\frac{1}{1.852}}
\end{bmatrix} \\
\label{eq:UKF-GSI_meas_prop_steps_c} \bm{\hat{y}}^{[k]}_{-} &= \sum_{i=0}^{2n} w_i^{(m)}\bm{\mathcal{Y}}^{[k]}_{i,-} \\
\label{eq:UKF-GSI_meas_prop_steps_d} \bm{\hat{P}}^{[k]}_{yy} &= \sum_{i=0}^{2n} w_i^{(c)}\left(\bm{\mathcal{Y}}^{[k]}_{i,-} - \bm{\hat{y}}^{[k]}_{-}\right)\left(\bm{\mathcal{Y}}^{[k]}_{i,-} - \bm{\hat{y}}^{[k]}_{-}\right)^{\top} + \bm{R}_h \\
\label{eq:UKF-GSI_meas_prop_steps_e} \bm{\hat{P}}^{[k]}_{xy} &= \sum_{i=0}^{2n} w_i^{(c)}\left(\bm{\mathcal{H}}^{[k]}_{i,-} - \bm{\hat{h}}^{[k]}_{-}\right)\left(\bm{\mathcal{Y}}^{[k]}_{i,-} - \bm{\hat{y}}^{[k]}_{-}\right)^{\top}
\end{align}
\end{subequations}
\normalsize

\noindent where $\bm{\mathcal{H}}^{[k]}_{-}$ is the set of head prior "sigma" points and $\hat{\bm{B}}^{[k]}$ is an approximated incidence matrix, 
computed considering the current hydraulic state of the network as:

\begin{equation}\label{eq:UKF(AW)GSI_incidence}
    \hat{b}^{[k]}_{oj}=\begin{cases}
    -1,& \hat{h}_i^{[k]}\geq  \hat{h}_j^{[k]}\;\; (\mathscr{e}_o = (\mathscr{v}_i,\mathscr{v}_j)\in \mathcal{E}) \\ 
    \hphantom{-}1,& \hat{h}_i^{[k]}< \hat{h}_j^{[k]} \;\; (\mathscr{e}_o = (\mathscr{v}_j,\mathscr{v}_i)\in \mathcal{E})\\ \hphantom{-}0,& (\mathscr{v}_i,\mathscr{v}_j)\notin \mathcal{E}\; \mbox{and} \;(\mathscr{v}_j,\mathscr{v}_i)\notin \mathcal{E}\end{cases}
\end{equation}

Moreover, $\bm{T}\in\mathbb{R}^{m\times m}$ is the resistance coefficient diagonal matrix, whose diagonal is formed by $\bm{\tau}$, defined in Section \ref{subsubsection:GSI}. The specific non-linear relation in \eqref{eq:UKF-GSI_meas_prop_steps_b} is derived from the non-linear flow-head relation, defined by expressions such as the Hazen-Williams equation, posed as:

\begin{equation}\label{eq:HW}
    \bm{q}^{[k]} = \left(\bm{T}^{-1}\bm{B}^{[k]}\bm{h}^{[k]}\right)^{\frac{1}{1.852}},
\end{equation} 

\noindent and the linear demand-flow relation, defined as a mass conservation equation at the network nodes, namely $\bm{c}^{[k]} = -\left(\hat{\bm{B}}^{[k]}_c\right)^{\top}\bm{q}^{[k]}$, where $\bm{c}^{[k]}$ is the nodal demand vector and $\hat{\bm{B}}^{[k]}_c$ is a submatrix of $\hat{\bm{B}}^{[k]}$ where only the columns corresponding to the nodes with ARMs are selected. Let us remark that the application of the non-linear part of $\mbox{\textbf{g}}$ to the prior "sigma" points $\bm{\mathcal{H}}^{[k]}_{-}$ in \eqref{eq:UKF-GSI_meas_prop_steps_b} implies the individual application of this equation to each "sigma" point, to then stack the obtained measurement "sigma" points horizontally.

Finally, the \textit{correction} step is done analogously to UKF:
\begin{subequations}
\label{eq:UKF-GSI_correc_steps}
\begin{align}
\label{eq:UKF_correc_steps_a}
 \bm{K}_h^{[k]} &= \bm{\hat{P}}^{[k]}_{xy}(\bm{\hat{P}}^{[k]}_{yy})^{-1} \\
\label{eq:UKF_correc_steps_b} \bm{\hat{h}}^{[k]} &= \bm{\hat{h}}^{[k]}_{-} + \bm{K}_h^{[k]}\left(\bm{z}_h - \bm{\hat{y}}^{[k]}_{-}\right) \\
\label{eq:UKF_correc_steps_c} \bm{\hat{P}}_h^{[k]} &= \bm{\hat{P}}^{[k]}_{h,-} - \bm{K}_h^{[k]}\bm{\hat{P}}^{[k]}_{yy}(\bm{K}_h^{[k]})^{\top} 
\end{align}
\end{subequations}

\noindent where $\bm{z}_h = \begin{bmatrix} \bm{h}_s^{\top} & \bm{c}_a^{\top} \end{bmatrix}^{\top}$ is the vector of actual measurements, and $\bm{c}_a$ is the vector of demand measurements. Again, note that $\bm{z}_h$ is constant during the complete run of the estimator.

\subsection{Improving flow estimation}\label{subsection:improving_flow} 

The presented UKF-based scheme for the integration of demand measurements showed its potential in \cite{RomeroBen2024b}. The proposed algorithm only considers the estimation of hydraulic heads, which play the role of the network state representative. Nevertheless, the proper estimation of other hydraulic variables such as flow could be of great importance for water utilities in network control operations. Therefore, we propose an extension to the methodology explained in Section \ref{subsection:including_demands} to integrate both head and flow estimation. Note that the estimation of multiple multi-dimensional variables has been studied in the past within the context of UKF. Notably, this method was extended in \cite{VanDerMerwe2004} for the estimation of both state and parameters, with a particular focus on learning-based schemes. This extension can be implemented in two different ways, depending on how the additional variable is integrated within the state estimation procedure. 

On the one hand, we can stack the variables to estimate a single joint state vector. In our case, this would translate into the incorporation of the flow vector into the state vector, which was already containing the head vector, leading to a composed state in the form of $\bm{x}^{[k]} = \begin{bmatrix}
        (\bm{h}^{[k]})^{\top} & (\bm{q}^{[k]})^{\top}
    \end{bmatrix}^{\top}$. 
On the other hand, the dual implementation uses separated head and flow state representations, and therefore implements two different estimators, one for heads and another for flows. When both approaches are compared, we can conclude that the joint strategy should be more precise because it considers cross-covariance terms in the error covariance matrix, while the dual approach neglects this information \cite{Nelson2000}. However, in practice, the difference in performance is not significant, and the dual version leads to other advantages, such as reduced computational cost. Thus, we propose a dual UKF/KF strategy, with the UKF estimator handling head reconstruction and a linear KF estimator helping in the computation of the flows. 

The flow KF-based estimator is composed of the steps presented in Section \ref{subsubsection:KF} with the \textit{prediction} step given by:
\begin{equation}\label{eq:KFFlow_pred}
\bm{\hat{q}}^{[k]}_{-} = \bm{\hat{q}}^{[k-1]}, \ \ 
\bm{\hat{P}}^{[k]}_{q,-} = \bm{\hat{P}}^{[k-1]}_q + \bm{Q}_q
\end{equation}
\noindent where $\bm{\hat{q}}^{[k-1]}$ is the estimated flow state and $\bm{\hat{P}}^{[k-1]}_q$ is the flow state error covariance matrix, with $\bm{Q}_q$ being the associated process noise covariance matrix. The flow state is simply maintained across iterations regarding prediction, as the purpose of this estimator is the synchronization of the measured flows with the virtual flows from the head UKF-based estimator.

The \textit{measurement update} step is given by:
\begin{subequations}
\label{eq:KFFlow_meas}
\begin{align}
    \label{eq:KFFlow_meas_a}
    \bm{K}_q^{[k]} &= \bm{\hat{P}}^{[k]}_{q,-}\bm{G}_q^{\top}\left(\bm{G}_q\bm{\hat{P}}^{[k]}_{q,-}\bm{G}_q^{\top} + \bm{R}_q\right)^{-1} \\
    \label{eq:KFFlow_meas_b}
    \bm{\hat{q}}^{[k]} &= \bm{\hat{q}}^{[k]}_{-} + \bm{K}_q^{[k]}\left(\bm{z}_q - \bm{G}_q\bm{\hat{q}}^{[k]}_{-}\right) \\
    \label{eq:KFFlow_meas_c} \bm{\hat{P}}^{[k]}_q &= \left(\bm{I}_n - \bm{K}_q^{[k]}\bm{G}_q\right)\bm{\hat{P}}^{[k]}_{q,-}
\end{align}
\end{subequations}

\noindent where $\bm{G}_q = \left[\bm{S}_q^{\top} \;\;\;\bm{I}_{|\mathcal{E}|}^{\top}\right]^{\top}$ is the measurement matrix, $\bm{K}_q^{[k]}$ is the Kalman gain, $\bm{R}_q$ is the associated measurement noise covariance matrix and $\bm{z}_q$ is the flow measurement vector. To integrate the head estimation from the UKF-based head estimator, a column vector
\begin{equation}\label{eq:z_q}
    \bm{z}_q = \begin{bmatrix} \bm{q}_s^\top & \left(\left(\bm{T}^{-1}\bm{\hat{B}}^{[k]}\bm{\hat{h}}^{[k]}\right)^{\frac{1}{1.852}}\right)^\top\end{bmatrix}^\top
\end{equation}
is constructed, where $\bm{q}_s$ denotes the measured flows, and the second entry computes the virtual flows derived from the estimated heads $\bm{\hat{h}}^{[k]}$ (note that $\bm{\hat{B}}^{[k]}$ must be computed to be consistent with the heads). 

\begin{remark}
    In the pipes with flow sensors, we would be imposing two values to the same state, that is, one from the real measurement and one from the virtual measurement. In order to prioritize the values from the network, a higher degree of confidence must be given to these measurements with the aid of the measurement noise covariance matrix. \eor
\end{remark}

Finally, regarding the UKF-based estimation process, the equations of the \textit{prediction} \eqref{eq:UKF-GSI_pred_steps}, \textit{measurement propagation} \eqref{eq:UKF-GSI_meas_prop_steps} and \textit{correction \eqref{eq:UKF-GSI_correc_steps}} are maintained except for the non-linear function application in \eqref{eq:UKF-GSI_meas_prop_steps_b}, which must be updated to consider the virtual measurements from the flow estimator:

\small
\begin{equation}\label{eq:dualUKF_meas_prop}
    \bm{\mathcal{Y}}^{[k]}_{-} = \mbox{\textbf{g}}(\bm{\mathcal{H}}^{[k]}_{-}) = \begin{bmatrix}
    \bm{S}\bm{\mathcal{H}}^{[k]}_{-} \\ -\left(\hat{\bm{B}}^{[k]}_c\right)^{\top}\left(\bm{T}^{-1}\hat{\bm{B}}^{[k]}\bm{\mathcal{H}}^{[k]}_{-}\right)^{\frac{1}{1.852}} \\
    \left(\bm{T}^{-1}\hat{\bm{B}}^{[k]}\bm{\mathcal{H}}^{[k]}\right)^{\frac{1}{1.852}}.
\end{bmatrix}
\end{equation}
\normalsize
Thus, the head estimator's measurement vector is defined as: 

\begin{equation}\label{eq:z_h}
    \bm{z}_h = \begin{bmatrix} \bm{h}_s^\top &
    \bm{c}_a^\top & \left(\bm{q}^{[k]}\right)^\top\end{bmatrix}^\top
\end{equation}

\noindent with $\bm{q}^{[k]}$ coming from the flow KF-based estimator.

\subsection{Methodology overview}

The head-flow estimation method, denoted as Dual UKF-(AW)GSI or D-UKF-(AW)GSI, is detailed in Algorithm \ref{alg:D-UKF-(AW)GSI}. 
\begin{algorithm}[t]
\small
\caption{Dual UKF-(AW)GSI.}\label{alg:D-UKF-(AW)GSI}
\begin{algorithmic}[1]
\REQUIRE $\bm{h}_{0}, \bm{F}, \bm{h}_s, \bm{c}_a, \bm{q}_s, \bm{S}, \bm{G}_q, \mathcal{C}, \bm{T}, \bm{P}_{0}^h, \bm{P}_{0}^q, \bm{Q}_h, \bm{R}_h, \bm{Q}_q, \bm{R}_q,$ $\bm{w}^{(m)}, \bm{w}^{(c)}, \eta, k_{D}$
\STATE Compute $\hat{\bm{B}}_{0}$ from $\bm{h}_{0}$ using \eqref{eq:UKF(AW)GSI_incidence}
\STATE Compute $\bm{q}_{0}$ from $\bm{h}_{0}, \hat{\bm{B}}_{0}$ and $\bm{T}$ using \eqref{eq:HW} 
\STATE Set $\bm{z}_h$ from $\bm{h}_s, \bm{c}_a, \bm{q}_0$ using \eqref{eq:z_h} and $\bm{z}_q$ from $\bm{q}_s, \bm{h}_0$ using \eqref{eq:z_q}
\STATE Initialize $\bm{\hat{h}}^{[0]} = \bm{h}_0$, $\bm{\hat{P}}^{[0]}_h = \bm{P}_{0}^h$, $\bm{\hat{q}}^{[0]} = \bm{q}_0$, $\bm{\hat{P}}^{[0]}_q = \bm{P}_{0}^q$, $k=1$ and $\delta_{conv} = $ False
\WHILE{$\delta_{conv} = $ False}
\STATE  Get $\bm{\hat{h}}^{[k]}_{-},\bm{\hat{P}}^{[k]}_{h,-}$ from $\bm{F}, \bm{\hat{h}}^{[k-1]},\bm{\hat{P}}^{[k-1]}_h,\bm{Q}_h$ using \eqref{eq:UKF-GSI_pred_steps} \\ 
\STATE Compute $\bm{\hat{q}}^{[k]}_{-},\bm{\hat{P}}^{[k]}_{q,-}$ from $\bm{\hat{q}}^{[k-1]},\bm{\hat{P}}_q^{[k-1]},\bm{Q}_q$ using \eqref{eq:KFFlow_pred}\\
\STATE Compute $\bm{\hat{B}}^{[k]}_{-}$ from $\bm{\hat{h}}^{[k]}_{-}$ using \eqref{eq:UKF(AW)GSI_incidence} \\
\STATE Extract $\bm{\hat{B}}^{[k]}_{c,-}$ from $\bm{\hat{B}}^{[k]}_{-}$ using $\mathcal{C}$ \\
\STATE Obtain $\bm{\mathcal{H}}^{[k]}_-$ from $\bm{\hat{h}}^{[k]}_{-},\bm{\hat{P}}^{[k]}_{h,-},\eta$ using \eqref{eq:UKF-GSI_meas_prop_steps_a}
\STATE Compute $\bm{\mathcal{Y}}^{[k]}_{-}$ from $\bm{\mathcal{H}}^{[k]}_-,\bm{S},\bm{T},\bm{\hat{B}}^{[k]}_{-},\bm{\hat{B}}^{[k]}_{c,-}$ using \eqref{eq:dualUKF_meas_prop}
\STATE Get $\bm{\hat{y}}^{[k]}_{-},\bm{\hat{P}}^{[k]}_{yy}, \bm{\hat{P}}^{[k]}_{xy}$ from $\bm{\mathcal{Y}}^{[k]}_{-},\bm{w}^{(m)}, \bm{w}^{(c)}, \bm{R}_h$ using \eqref{eq:UKF-GSI_meas_prop_steps_c}-\eqref{eq:UKF-GSI_meas_prop_steps_e} \\
\STATE Compute $\bm{\hat{h}}^{[k]}, \bm{\hat{P}}^{[k]}_h$ from $\bm{\hat{h}}^{[k]}_{-},\bm{\hat{P}}^{[k]}_{h,-}, \bm{z}_h, \bm{\hat{y}}^{[k]}_{-}, \bm{\hat{P}}^{[k]}_{yy}, \bm{\hat{P}}^{[k]}_{xy}$ using \eqref{eq:UKF-GSI_correc_steps}\\
\STATE Get $\bm{\hat{q}}^{[k]}, \bm{\hat{P}}^{[k]}_q$ from $\bm{\hat{q}}^{[k]}_{-},\bm{\hat{P}}^{[k]}_{q,-}, \bm{z}_q, \bm{G}_q, \bm{R}_q$ using \eqref{eq:KFFlow_meas}
\IF{$\mbox{mod}(k,k_{D}) = 0$}
     \STATE Compute $\bm{\hat{B}}^{[k]}$ from $\bm{\hat{h}}^{[k]}$ using \eqref{eq:UKF(AW)GSI_incidence}
     \STATE Set $\bm{z}_h$ from $\bm{h}_s, \bm{c}_a, \bm{\hat{q}}^{[k]}$ using \eqref{eq:z_h} and $\bm{z}_q$ from $\bm{q}_s, \bm{\hat{h}}^{[k]}$ using \eqref{eq:z_q}
\ENDIF
\STATE $\delta_{conv} =$ convergence\_criteria($\bm{\hat{h}}^{[k]},\bm{\hat{h}}^{[k-1]},\bm{\hat{q}}^{[k]},\bm{\hat{q}}^{[k-1]}$)
\STATE $k = k + 1$
\ENDWHILE
\RETURN $\bm{h}^{[k]}, \bm{q}^{[k]}$
\end{algorithmic}
\end{algorithm}
First,
in order to define the initial measurement vectors $\bm{z}_h$ and $\bm{z}_q$ (step 3), the flow associated to initial heads $\bm{q}_0$ needs to be computed (steps 1-2). After the initialization of state and covariance matrix for both estimators, iteration counter and convergence flag (step 4), the estimation loop begins (step 5). The \textit{prediction} process (steps 6-7) is analogue for both head and flow estimators. Then, the \textit{measurement propagation} step of the head UKF (steps 7-12) is required to apply the non-linear measurement and get the variables involved in the \textit{correction} step (step 13). The flow KF-based estimation ends at the current iteration with the \textit{measurement update} (step 14), unless the condition related to the virtual measurement update is fulfilled (step 15). In this case, each measurement vector is updated considering the state of the other estimator (steps 16-17). Finally, the convergence criteria is checked, finalizing the process if the conditions are met (step 19).


\begin{remark}
    The UKF-(AW)GSI method can also be represented by Algorithm \ref{alg:D-UKF-(AW)GSI} if the processes related to the flow KF are disabled (steps 1-2, 7, 14) or modified (steps 3-4, 19), the virtual measurements update are removed (steps 15-18) and the application of the non-linear function (step 11) is modified to use \eqref{eq:UKF-GSI_meas_prop_steps_b} instead of \eqref{eq:dualUKF_meas_prop}. \eor
\end{remark}

\begin{remark}
    Regarding computational complexity, Algorithm \ref{alg:D-UKF-(AW)GSI} can be analyzed considering the head and flow estimators separately. The head estimation UKF is primarily influenced by the matrix inversion and square-root operations. Both of them are usually solved through methods such as Cholesky decomposition, with a complexity of $\mathcal{O}(\mathfrak{n}^3)$, with $\mathfrak{n}$ being the size of the matrix. In this case the inversed matrix is $\bm{\hat{P}}_{yy}^{[k]}$, which has a dimension of $n_{\bm{\hat{P}}_{yy}^h} = n_s + |\mathcal{C}| + m$, while the square-root computation is performed on $\bm{\hat{P}}_{h,-}^{[k]}$, with a size $n_{\bm{\hat{P}}_{h,-}^{[k]}} = n$. Additionally, the complexity of the flow estimation KF is mainly governed by the matrix inversion, which is required for a matrix with the size of $\bm{R}_q$, i.e., $n_{\bm{R}q} = n_q + m$, where $n_q$ is the dimension of the flow measurements vector. Thus, the overall complexity of D-UKF-(AW)GSI is $\mathcal{O}(\mathfrak{n}^3)$, with the computation cost depending on the actual values of the network size and sensorization properties.  
    \eor
\end{remark}



\section{Case study}\label{section:case_study}
Hereinafter, the L-TOWN benchmark from the Battle of
Leakage Detection and Isolation Methods 2020 - BattLeDIM2020 \cite{Vrachimis2022}, shown in Fig.~\ref{fig:L-TOWN}, is used to test the state estimation and leak localization algorithms. The network has 782 junctions, 905 pipes, 2 reservoirs or water inlets and one tank. The water utility divides the WDN in three areas, depending on the elevation of the junctions:
    area A (green in Fig.~\ref{fig:L-TOWN}) has 655 nodes, with 29 pressure sensors,  the area is fed through two reservoirs; 
    area B (blue in Fig.~\ref{fig:L-TOWN}) has 31 nodes, with only one pressure sensor;  it is connected to Area A and receives water through a pressure reduction valve (PRV);
    area C (yellow in Fig.~\ref{fig:L-TOWN}) has 92 nodes, with only 3 pressure sensors but 82 Automated Meter Reading (AMR) consumption sensors; water is fed to it through a tank, filled from Area A. 
\begin{figure}[t]
\centering
\includegraphics[width=\linewidth]{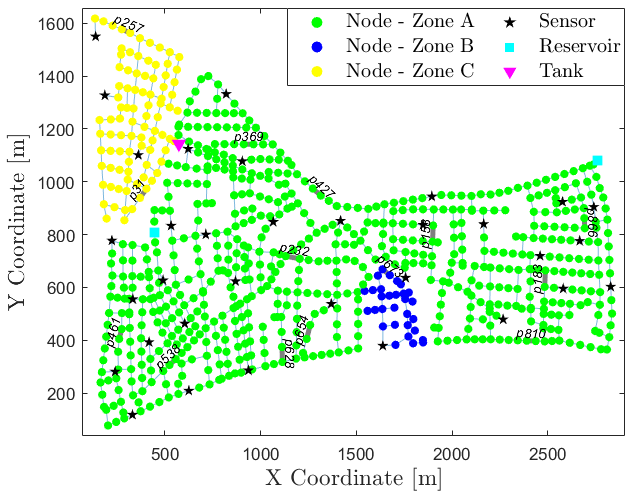}
\caption{Schematic representation of the L-TOWN water network. The different Areas are represented through different colors (Area A - green, Area B - blue, Area C - yellow), and the reservoirs, tank and sensors are indicated through squared, triangular and star markers. The leaks that occurred in 2018 are labelled and marked over the corresponding pipe.}
\label{fig:L-TOWN}
\end{figure}
BattLeDIM2020 was a leak detection and localization competition in which several teams analyzed the leaks along the year 2019 through SCADA measurements from pressure, demand and flow sensors. The organization provided a calibration dataset from the previous year (2018) with several a priori located leaks. The pair of datatsets provided by the competition organizers contain 33 pressure readings and 82 consumption measurements, together with 3 flow sensors (at the exit of reservoirs and tank) and 1 level sensor at the tank.

\section{Results \& Discussion\protect\footnotemark}\label{section:results}

\footnotetext{Code and data: \url{https://github.com/luisromeroben/D-UKF-AW-GSI}}

To assess the performance of Dual UKF-(AW)GSI, we compare it with existing interpolation methods. Numerical results are analyzed from two perspectives: the accuracy of the estimation process and the leak localization performance. To carry out this comparison, the leaks from the 2018 dataset are used, to reduce the number of overlapping leaks and thus the complexity for the leak localization procedure.
\begin{figure}[t]
\centering
\includegraphics[width=.9\linewidth]{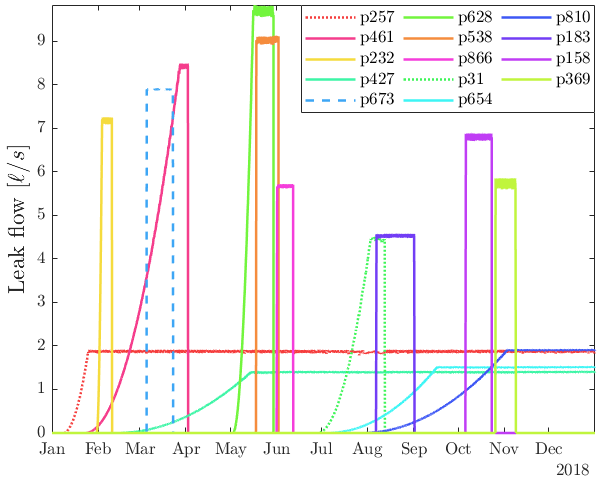}
\caption{Evolution of the leak rate of the leakage occurred during 2018. Continuous lines indicate leaks that appeared in Area A, whereas dotted lines indicate leaks in Area B and dashed lines show leaks in Area C.}
\label{fig:leaks_2018}
\end{figure}
\subsection{Estimation}

The evaluation of the estimation capabilities of the proposed UKF-based methods is performed by comparing the reconstruction error of these methods to the baselines from previous works, i.e., GSI and AW-GSI. Since Area A is the largest within the network and contains the vast majority of the leaks occurring in 2018, it will be subsequently used to assess the performance of our methodology. Nonetheless, the original dataset has no demand measurements from Area A. To run our algorithms, a re-calculation of the 2018 dataset is done through the Dataset Generator\footnote{\url{https://github.com/KIOS-Research/BattLeDIM}} provided by the organizers. Thus, we can recreate the conditions of the original 2018 dataset while adding AMRs in Area A. Several key aspects must be considered.
    First, the water inlets are isolated from the rest of Area A through a PRV, which maintains a constant head of 75 m at the node immediately downstream, while reservoirs have a water height of 100 m. This abrupt discrepancy of hydraulic heads between reservoirs and the rest of the area degrades the UKF performance. To address this, the network graph structure is adjusted by removing the reservoirs, and setting the nodes immediately downstream of the PRVs as the new "water inlets", with a head of 76 m, to ensure (as required from the GSI / AW-GSI perspective) that their head values are the highest in the network. 
    Second,
    in the presented experiments, we consider that 100 AMRs exist in Area A. They are placed through a model-free sensor placement methodology \cite{RomeroBen2022sp}.
    And finally, due to the higher computational cost of the UKF (in comparison with non-UKF methods), some measures are considered to reduce computation time. Specifically, 100 leaks are considered out of all potential network leaks. They are selected using the same method employed to get the AMR locations, but ensuring that these locations are not selected as leak nodes too. Moreover, only one time instant is considered for each leak, highly reducing the associated computations.

Estimation results are presented in Fig.~\ref{fig:head_estimation} and Fig.~\ref{fig:flow_estimation}. Each figure represents the distribution of the estimation error for their respective hydraulic variable, that is, head and flow. The error is computed through the root-squared-mean error,
$RMSE(\bm x, \hat{\bm x}) = \sqrt{\frac{1}{n} \sum_{i=1}^{n} (x_i - \hat{x}_i)^2}$, where $\bm{x}$ is a generic vector.
The results clearly show how GSI and AW-GSI lead to a similar performance, with the latter producing a better flow estimation. A great improvement is achieved for both head and flow estimation when demand measurements are integrated through the UKF-based approach. Note that D-UKF-(AW)GSI performs similarly to UKF-(AW)GSI in terms of head estimation. Nonetheless, the dual approach leads to an enhanced flow estimation.
\begin{figure}[t]
\centering
\includegraphics[width=\linewidth]{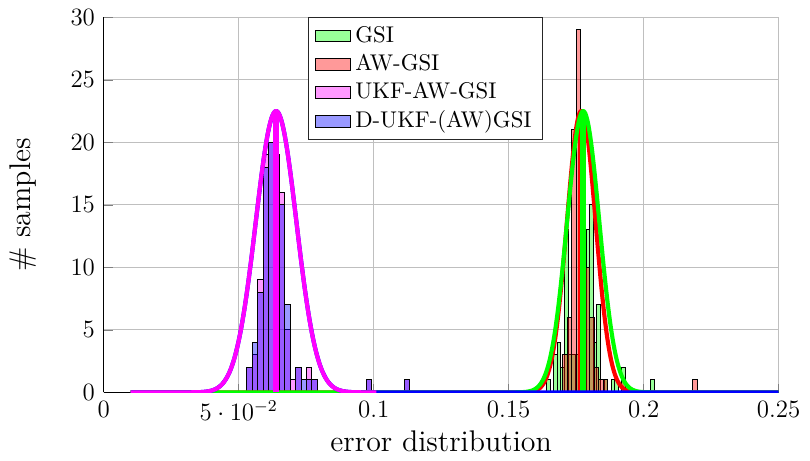}
\caption{Head estimation RMSE of GSI, AW-GSI, UKF-(AW)GSI and D-UKF-(AW)GSI for 100 considered leak scenarios.}
\label{fig:head_estimation}
\end{figure}
\begin{figure}[t]
\centering
\includegraphics[width=\linewidth]{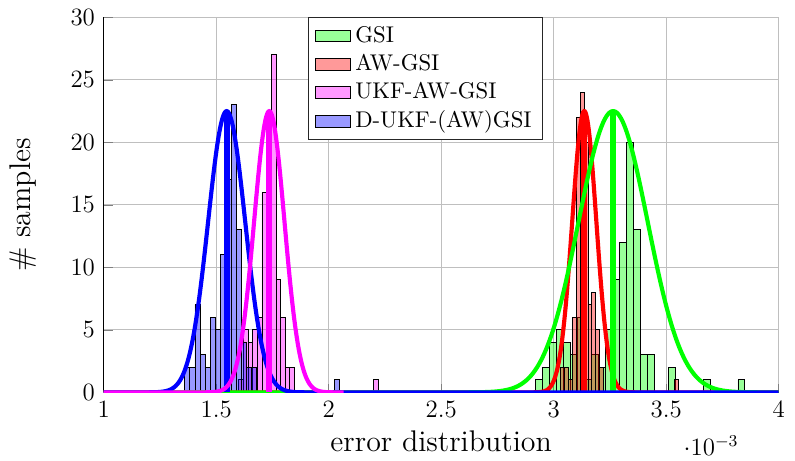}
\caption{Flow estimation RMSE of GSI, AW-GSI, UKF-(AW)GSI and D-UKF-(AW)GSI for 100 considered leak scenarios.}
\label{fig:flow_estimation}
\end{figure}
The results in these figures is complemented with the numerical results in Table \ref{table:estimation_results} where the mean and standard deviation of the RMSE estimation error for all the considered leaks is presented.
\begin{table}[t] 
    \centering
    \resizebox{\linewidth}{!}{
    \begin{tabular}{cc|cccc}
        \multicolumn{2}{c|}{\multirow{2}{*}{\textbf{Case}}} & \multicolumn{4}{c}{\textbf{Estimation RMSE ($\mu\pm\sigma$)}} \\
         \multicolumn{2}{c|}{}& (A) & (B) & (C) & (D)  \\
        \hline
        \multirow{2}{*}{(I)} & $h$ (cm) &
        $82.19 \pm 5.73$ & $81.12\pm5.59$ & $48.67\pm5.74$ & $48.62\pm5.65$ \\
        & $q$ $(\ell/s)$ & 
        $13.43\pm0.48$ & $13.15\pm0.29$ & $4.65\pm1.63$ & $3.54\pm0.11$ \\
        \hline
        \multirow{2}{*}{(II)} & $h$ (cm) &
        $17.76\pm0.60$ & $17.69\pm0.53$ & $6.39\pm0.75$ & $6.39\pm0.75$ \\
        & $q$ $(\ell/s)$& 
        $3.26\pm0.15$ & $3.14\pm0.05$ & $1.73\pm0.07$ & $1.55\pm0.08$ \\
    \end{tabular}%
    }
    \vspace{3pt}
    \caption{Estimation performance in Modena (I) and L-TOWN (II) for GSI (A), AW-GSI (B), UKF-(AW)GSI (C) and D-UKF-(AW)GSI (D).}
    \label{table:estimation_results}
\end{table}
Apart from the L-TOWN benchmark, the well-known Modena case study \cite{Bragalli2012,Alves2021} is added to enrich the results and show the consistency of the improvements, considering that this benchmark was already used in \cite{RomeroBen2024b}, where the properties of the case study and the evaluated data are introduced. For Modena, all the leaks are analyzed, and as in the case of L-TOWN, only one time instant is considered to reduce computation time.

\subsection{Localization}




\noindent\textbf{Area A.}
The performance of localization experiments in Area A required several adjustments. 
    First and again, the water inlets are considered at the nodes just after the PRVs associated to each reservoir. Moreover, the same 100 AMRs are considered for these experiments.
    Secondly, for each leak, a period of 24 hours is selected, with a sampling of 30 minutes. The values are chosen such as to minimize leak overlap, as per Fig.~\ref{fig:leaks_2018}.
    And finally, only medium-to-large leaks are studied, avoiding background leaks (specifically, \textit{p427}, \textit{p654} and \textit{p810}), whose effect is negligible and masked by the main leaks.  

The UKF's 
        SUT parameters are $\alpha=10^{-3}$ and $\beta = 2$, both typical in the literature \cite{Julier1997}.
        The process and measurement noise covariance matrices are $\bm{Q} = \bm{I}_{n^*}$ and $\bm{R} = 10^{-4}\bm{I}_{n_s+n_c}$, where $n^*$ is the number of nodes in Area A neglecting the original reservoirs, and $n_c$ is the number of AMRs in Area A. In this way, higher degree of certainty is given to the measurements.
        The initial state is computed through AW-GSI and the state error covariance matrix is given as $\bm{P}_0 = \bm{I}_{n^*}$.

Tables \ref{table:UKFAWGSIvsAWGSI_area} and \ref{table:UKFAWGSIvsAWGSI_node}
present performance metrics for each leak, using a set of key performance indicators (KPI) for area-level and node-level localization, respectively. Area-level localization provides a set of network elements as the possible leak locations, defining an area where the potential leak is hidden. Node-level localization provides a single network element as the leak localization result.
The following KPIs are used in this work.
    $b_c$, a boolean variable indicating if the leak has been included within a set of high-likelihood candidates by the localization algorithm. This set is defined by thresholding the Leak Candidate Selection Method (LCSM) metric, which in this case must be $\ge 0.7$ for the corresponding pipe\footnote{The LCSM metric is computed for nodes. The associated LCSM metric to each pipe is given by the average LCSM metric of the nodes which are endpoints of the corresponding pipe.} to be included in the high-likelihood candidates set. The LCSM metric is normalized in the [0,1] range. 
    $\bar{d}^{c2l}$ denotes the weighted averaged pipe distance (in meters) from the pipes in the high-likelihood candidates set to the actual leaky pipe. The distance between two distinct pipes is
    $d^{c2l}_{k} = \frac{1}{4}\sum_{u=i_k}^{j_k}\sum_{v=i_l}^{j_l} \mbox{shortest\_path}(u,v)$
    where $\mathscr{e}_{k} = \mathscr{e}_{i_k j_k} = (\mathscr{v}_{i_k},\mathscr{v}_{j_k})$ is the $k$-$th$ candidate pipe, $\mathscr{e}_{l} = \mathscr{e}_{i_l j_l} = (\mathscr{v}_{i_l},\mathscr{v}_{j_l})$ is the leaky pipe, and shortest\_path($\cdot$) is a function computing the shortest-path distance between two nodes. 
    With this, the actual KPI is computed as $\bar{d}^{c2l} = \left(\bm{\tilde{w}}^{LM}\right)^{\top}\bm{d}^{c2l}$
    where $\bm{\tilde{w}}^{LM} = \frac{\bm{w}^{LM}}{\sum_{k=1}^{n_{cands}} w_k^{LM}}$ is a column vector with the localization metric values associated to the candidates, and $n_{cands}$ is the number of candidate pipes. 
    $\bar{p}^{c2l}$ is computed in the same way as $\bar{d}^{c2l}$, but with the distance expressed in number of pipes. 
    $\rho_c$ gives a measure of the search area (with respect to the specific area of the WDN, namely Area A, B or C). It is calculated as the number of high-likelihood candidate pipes per number of pipes within the area. 
    $d_{best}^{c2l}$ expresses the pipe distance (in meters) from the actual leak to the candidate with the highest associated LCSM metric. Thus, the node-level performance of the method can be analyzed, with this KPI representing the actual node-level localization error.
    $p_{best}^{c2l}$ denotes the same metric as $d_{best}^{c2l}$, but expressed in number of pipes.
\begin{table}[t] 
    \centering
    \begin{tabular}{ccccccccc}
         & \multicolumn{2}{c}{$\bm{b_{c}}$} & \multicolumn{2}{c}{$\bm{\bar{d}^{c2l}}$ \textbf{(m)}} & \multicolumn{2}{c}{$\bm{\bar{p}^{c2l}}$ (\textbf{\# pipes)}} & \multicolumn{2}{c}{$\bm{\rho_c}$ \textbf{(\%)}} \\
         & (I) & (II) &  (I) & (II) &  (I) & (II) &  (I) & (II) \\
        \hline
         p461 & 
        1 & 1 & 
        227.58 & 238.47 & 
        4.92 & 5.18 & 
        7.61 & 7.48  \\
         p232 & 
        1 & 1 &  
        385.06 & 395.09 & 
        8.02 & 8.45 & 
        16.80 & 21.39 \\
         p628 & 
        1 & 1 & 
        339.83 & 359.52 & 
        6.65 & 7.15 & 
        14.04 & 19.16  \\
         p538 & 
        0 & 0 & 
        477.48 & 495.14 & 
        9.95 & 10.25 & 
        11.42 & 11.55  \\
         p866 & 
        0 & 1 &  
        236.76 & 232.41 & 
        5.07 & 4.95 & 
        5.25 & 5.51  \\
         p183 & 
        1 & 1 &
        308.14 & 340.25 &
        7.18 & 7.84 &
        13.12 & 17.72 \\
         p158 & 
        1 & 1 & 
        176.41 & 174.93 & 
        3.42 & 3.38 & 
        6.56 & 6.69 \\
         p369 & 
        0 & 0 & 
        448.14 & 490.97 & 
        9.24 & 10.19 & 
        4.07 & 9.45 \\
        \bottomrule
        \multicolumn{3}{c}{AVERAGE} & 
        324.93 & 340.85 &  
        6.81 & 7.18 & 
        9.86 & 12.37  \\
    \end{tabular}
    \caption{Area A area-level localization: AW-GSI (I), UKF-AWGSI (II)}
    \label{table:UKFAWGSIvsAWGSI_area}
\end{table} 
\begin{table}[t] 
    \centering
    \begin{tabular}{ccccccc}
         & \multicolumn{2}{c}{$\bm{d^{c2l}_{best}}$ \textbf{(m)}} & \multicolumn{2}{c}{$\bm{p}^{c2l}_{best}$ (\textbf{\# pipes)}} \\
         & (I) & (II) & (I) & (II) \\
        \hline
         p461 &
        334.00 & 293.53 &
        8.00 & 7.00  \\
         p232 & 
        516.77 & 427.41 &
        11.00 & 9.00 \\
         p628 & 
        401.94 & 375.74 &
        8.00 & 7.00   \\
         p538 & 
        465.20 & 320.16 &
        9.00 & 6.00  \\
         p866 & 
        240.78 & 214.40 &
        6.00 & 5.00 \\
         p183 & 
        181.53 & 211.78 &
        5.00 & 5.00  \\
         p158 & 
        44.32 & 44.32 &
        1.00 & 1.00 \\
         p369 & 
        502.74 & 406.81 &
        10.00 & 8.00 \\
        \bottomrule
        AVERAGE & 
        335.91 & 286.77 &
        7.25 & 6.00 \\
    \end{tabular}
    \caption{Area A node-level localization: AW-GSI (I), UKF-AWGSI (II)}
    \label{table:UKFAWGSIvsAWGSI_node}
\end{table} 

It is important to note that the results in Tables \ref{table:UKFAWGSIvsAWGSI_area} and \ref{table:UKFAWGSIvsAWGSI_node} compare AW-GSI with UKF-(AW)GSI. The non-dual version is selected here because its head estimation performance is nearly identical to that of D-UKF-(AW)GSI (whose main advantage lies in flow estimation), while the computing cost is lower. Given the extensive number of leaks and time instants analyzed, UKF-(AW)GSI was chosen to minimize computational demands while preserving performance.

The results allow several conclusions.
    Regarding area-level localization, the performances are comparable, although AW-GSI shows slightly better results. In terms of $b_c$, UKF-(AW)GSI provides a satisfactory result in an additional scenario in comparison to AW-GSI, namely \textit{p866}. Regarding the candidate-to-leak distance (both in meters and number of pipes), it is higher for the UKF-based strategy, with around a 5\% increase. However, note that search area is also higher in UKF-(AW)GSI ($\sim$25\% increase). Thus, this increase of the localization error could be explained by the inclusion of extra nodes in the high-likelihood candidate set.
    Also, the node-level localization shows the main advantage of improving the estimation process. The distance from the best candidate to the leak is reduced around a 15\% when the demand information is integrated through the UKF-based strategy. 

These results demonstrate the importance of making the most of all the available sources of information within the network. UKF-(AW)GSI allows to integrate additional types of measurements, leading to a better node-level performance (which is one of the typical drawbacks of data-driven schemes) while maintaining a comparable area-level performance.
\\

\noindent\textbf{Area C.} Regarding this area, 
the BattLeDIM2020 datasets actually include AMR data. Thus, the method can be tested without the need of computing additional data from the leaks of 2018. The reduced size of Area C allows us to increase the number of analyzed samples, to consider a 5-minute sensor sampling and to run the main UKF loop for 200 iterations.
\begin{table}[t] 
    \centering
    \begin{tabular}{ccccccccc}
         & \multicolumn{2}{c}{$\bm{b_{c}}$} & \multicolumn{2}{c}{$\bm{\bar{d}^{c2l}}$ \textbf{(m)}} & \multicolumn{2}{c}{$\bm{\bar{p}^{c2l}}$ (\textbf{\# pipes)}} & \multicolumn{2}{c}{$\bm{\rho_c}$ \textbf{(\%)}} \\
         & (I) & (II) &  (I) & (II) &  (I) & (II) &  (I) & (II) \\
        \hline
         p257 & 
        1 & 1 & 
        176.12 & 152.86 & 
        3.28 & 2.83 & 
        25.69 & 20.18  \\
         p31 & 
        0 & 0 &  
        618.22 & 577.15 & 
        12.15 & 11.28 & 
        14.68 & 9.17 \\
    \end{tabular}
    \caption{Area C area-level localization: AW-GSI (I), UKF-AWGSI (II)}
    \label{table:UKFAWGSIvsAWGSI_area_AreaC}
\end{table} 
\begin{table}[t] 
    \centering
    \begin{tabular}{ccccccc}
         & \multicolumn{2}{c}{$\bm{d^{c2l}_{best}}$ \textbf{(m)}} & \multicolumn{2}{c}{$\bm{p}^{c2l}_{best}$ (\textbf{\# pipes)}} \\
         & (I) & (II) & (I) & (II) \\
        \hline
         p257 &
        105.26 & 222.27 &
        2.00 & 4.00  \\
         p31 & 
        623.15 & 194.22 &
        12.00 & 4.00 \\
    \end{tabular}
    \caption{Area C node-level localization: AW-GSI (I), UKF-AWGSI (II)}
    \label{table:UKFAWGSIvsAWGSI_node_AreaC}
\end{table} 

The obtained results, shown in Tables \ref{table:UKFAWGSIvsAWGSI_area_AreaC} and~\ref{table:UKFAWGSIvsAWGSI_node_AreaC}, allow the following conclusions.
    For leak \textit{p257}, both methods provide satisfactory performance, although UKF-(AW)GSI-LCSM exhibits some node-level accuracy degradation. However, the UKF-based method behaves slightly better when considering area-level performance. This indicates that the method selects a proper set of candidates, but that the rank of these candidates in terms of their LCSM metric is degraded with respect to AW-GSI-LCSM. Nonetheless, leak \textit{p257} is a background leak, thus the results are less reliable than for medium-to-large leaks.
    Also, notably the difference in performance happens when a multi-leak scenario is addressed, considering that leak \textit{p31} occurs while leak \textit{p257} is still active. In this scenario, a slight improvement is observed in terms of area-level performance when using the UKF-based method. Besides, node-level performance is largely improved, with a distance error of 12 pipes for AW-GSI-LCSM and only 4 pipes for UKF-(AW)GSI-LCSM. In terms of $d_{best}^{c2l}$, this implies a 67\% reduction of the distance between the best candidate and the leak. 

In order to completely analyze the multi-leak event, graphical localization results are shown for AW-GSI-LCSM and UKF-(AW)GSI-LCSM in Fig.~\ref{fig:leak_localization_AreaC_AWGSI} and Fig.~\ref{fig:leak_localization_AreaC_UKFAWGSI}, respectively. 
\begin{figure}[t]
     \centering
    \includegraphics[width=\linewidth]{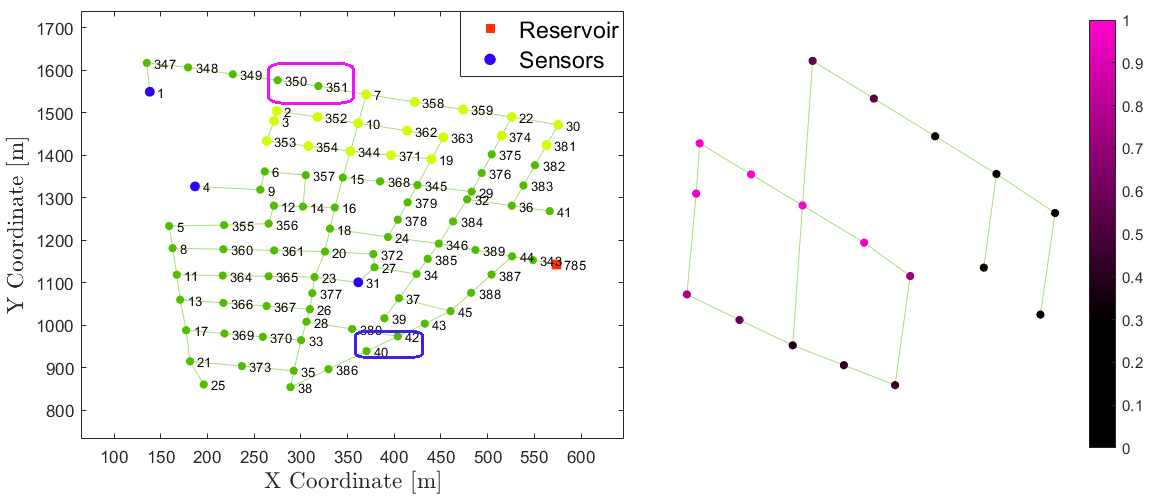}
        \caption{AW-GSI-LCSM: Area C localization for leak at pipe 31.}
    \label{fig:leak_localization_AreaC_AWGSI}
\end{figure}
\begin{figure}[t]
    \centering
    \includegraphics[width=\linewidth]{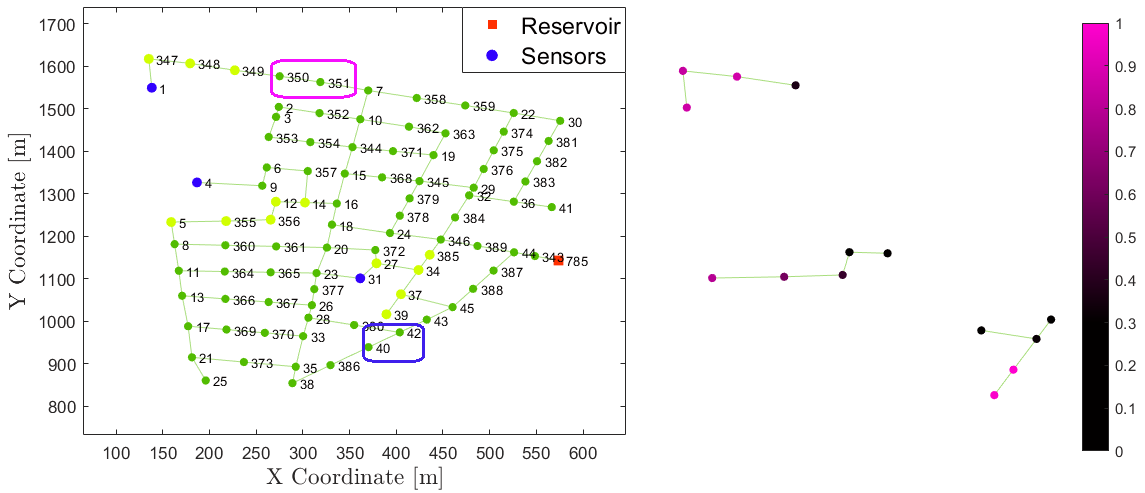}
    \caption{UKF-(AW)GSI-LCSM: Area C localization for leak at pipe 31.}
    \label{fig:leak_localization_AreaC_UKFAWGSI}
\end{figure}
Each figure is divided in two subplots: the left plot corresponds to the general view, and the right plot corresponds to the detailed view. The general view shows a graph of Area C, marking the pressure sensors with blue dots and the water inlet (tank) with red squares. LCSM candidates are indicated using lighter nodes (in comparison to the rest of the nodes). Occurring leaks are enclosed by magenta (leak \textit{p257}) and blue (leak \textit{p31}) lines. The detailed perspective only shows the candidate nodes, using a colormap (indicated through the colorbar in the right) to denote the likelihood of each node as the leak source, with light magenta indicating higher probability. 

Fig.~\ref{fig:leak_localization_AreaC_AWGSI} shows that AW-GSI-LCSM is not capable to correctly localize the larger leak (\textit{p31}), yielding an intermediate result between this leak and the background leak in \textit{p257}. Instead, Fig.~\ref{fig:leak_localization_AreaC_UKFAWGSI} shows that the localization using UKF-(AW)GSI-LCSM is successful: the most likely leak origins are located close to \textit{p31}. Additionally, the background leak is at only 3 pipes from the second most likely leak area, allowing to conclude that the multi-leak scenario is properly solved. This performance highlights even more the power of exploiting extra sources of information such as nodal consumptions, if they are available.

\section{Conclusions \& Future work}\label{section:conclusions}

This article presents a dual hydraulic state estimation method to retrieve both head and flow vectors from an initial guess and different types of measurements, with the aim of improving leak localization. The initial guess can be obtained from interpolation methods such as GSI or AW-GSI, and measurements can include pressure, flow and demand, which are fused through a UKF-based strategy.

The approach is mainly tested in the well-known open-source L-TOWN benchmark from BattLeDIM2020. Both estimation and localization are improved with respect to previous state-of-the-art state estimation/leak localization methods. Regarding localization, the UKF-based methods show promising performance handling multi-leak events, which are typically challenging. The estimation capabilities are also evaluated in the case study of Modena, aiding to show the consistency of the improvement over different networks.

Several improvements to the method can be explored in future works. First, the UKF's power to handle non-linear relationships could be leveraged to integrate the presence of active network elements, such as valves or pumps. Then, the prediction function of both head and flow estimators may be revisited in order to improve the transition between iterations. Moreover, other sensor fusion methods can be evaluated with the aim of enhancing estimation, such as different non-linear KF-based methods or Factor Graph Optimization (FGO).

\bibliographystyle{IEEEtran}
\bibliography{references}













\begin{IEEEbiography}[{\includegraphics[width=1in,height=1.25in,clip,keepaspectratio]{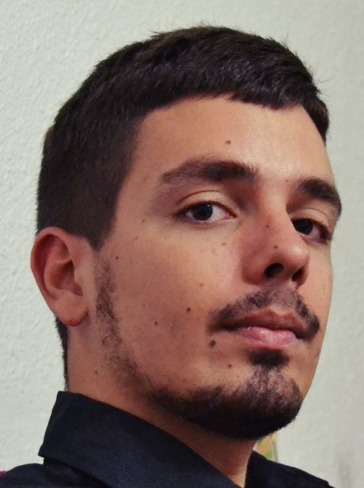}}]{Luis Romero-Ben}
is currently a PhD student and researcher at the Automatic Control group at the Institut de Robòtica e Informàtica Industrial (CSIC-UPC), Spain. His interests are focused on the development and application of data-driven methodologies for the control and monitoring of water distribution networks and urban drainage systems.
\end{IEEEbiography}

\begin{IEEEbiography}[{\includegraphics[width=1in,height=1.25in,clip,keepaspectratio]{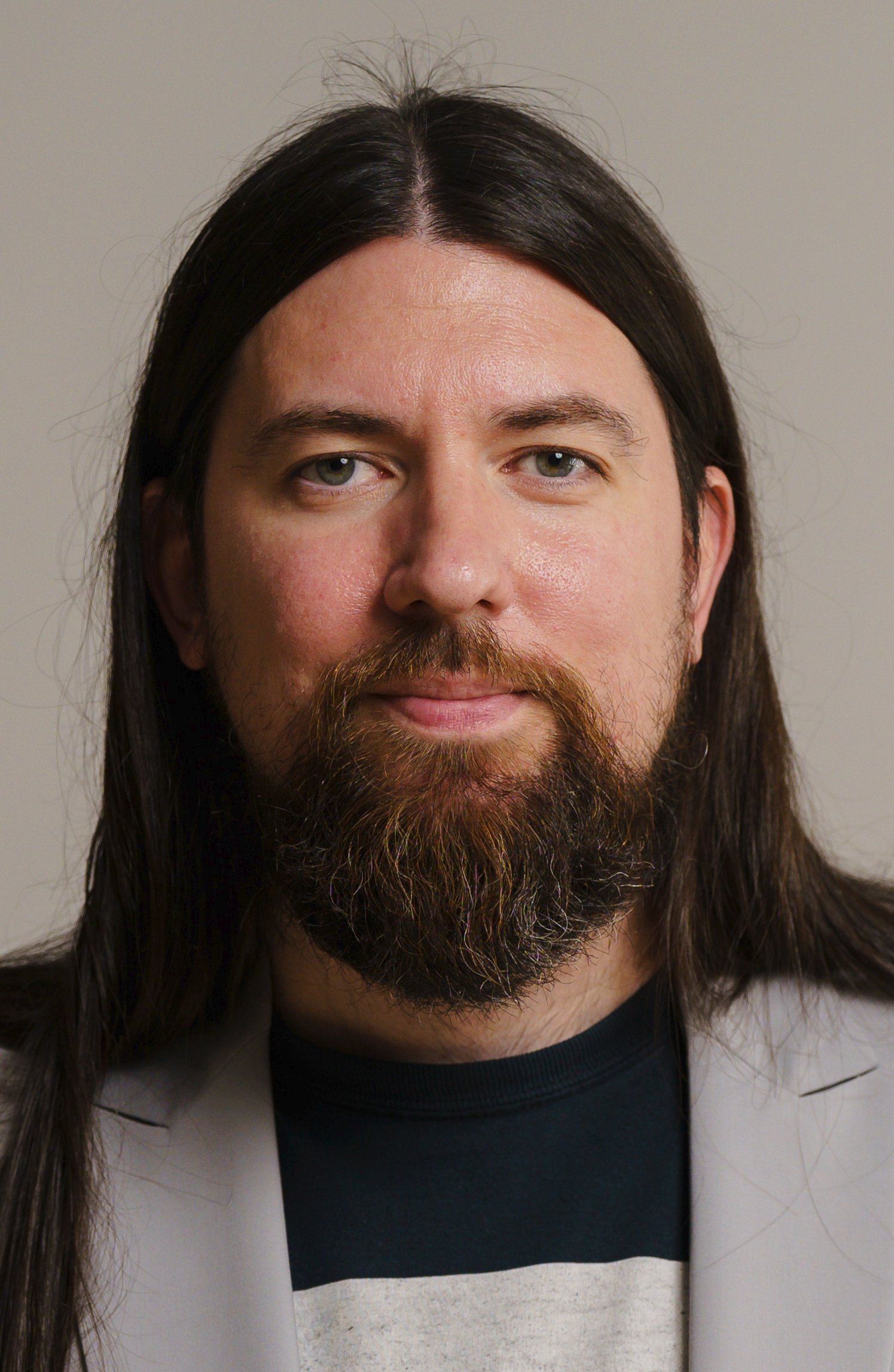}}]{Paul Irofti}
is an Associate Professor
within the Computer Science Department
of the Faculty of Mathematics and Computer Science
at the University of Bucharest
He is the co-author of the book “Dictionary Learning Algorithms and Applications” (Springer 2018)
awarded by the Romanian Academy.
He is PhD in Systems Engineering at the Politehnica University of
Bucharest since 2016.
His interests are anomaly detection, signal processing, 
numerical algorithms and optimization.

\end{IEEEbiography}

\begin{IEEEbiography}[{\includegraphics[width=1in,height=1.25in,clip,keepaspectratio]{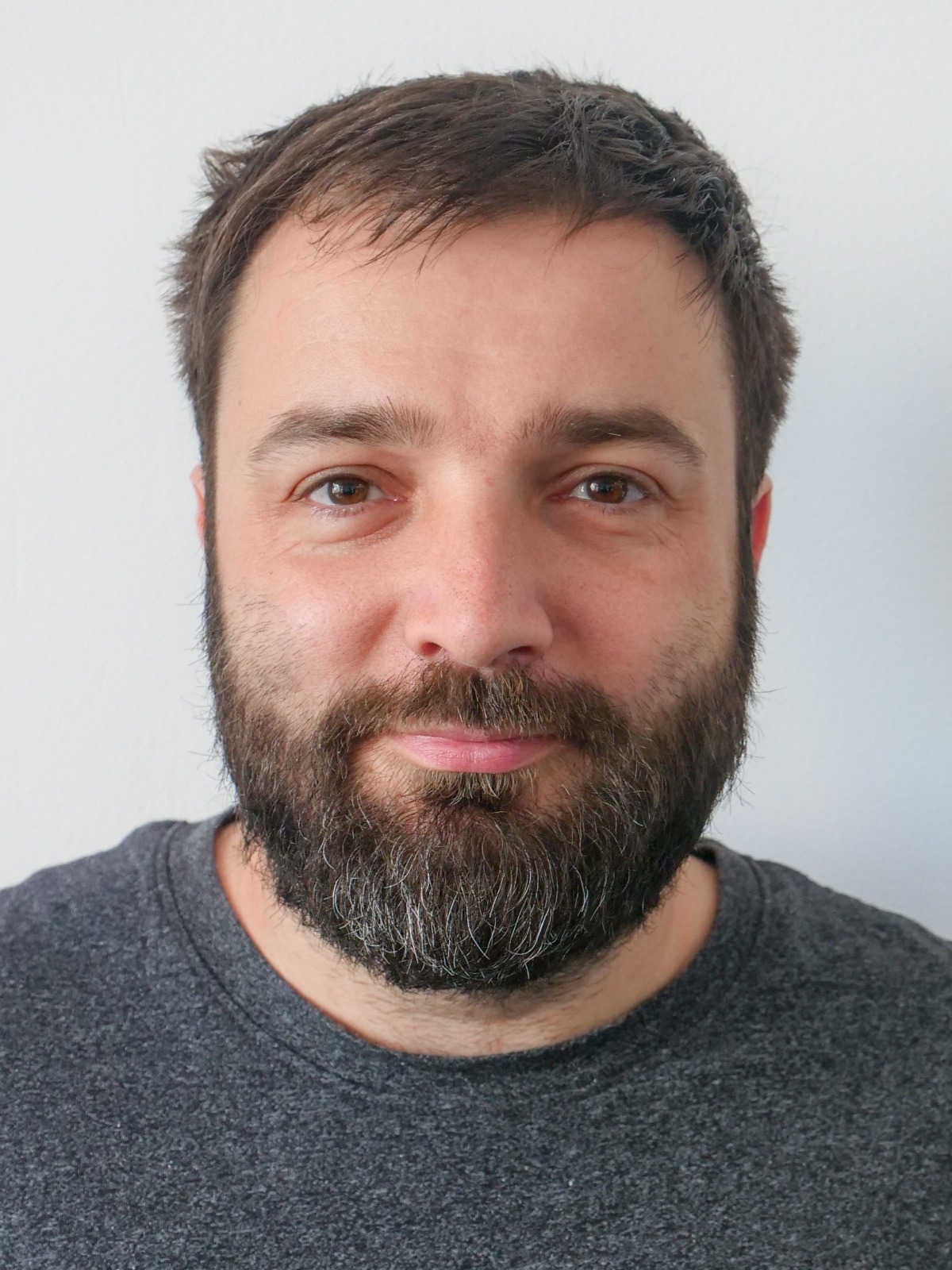}}]{Florin Stoican}
 is Professor in the department of Automatic Control and Systems Engineering, Politehnica University of Bucharest. He obtained his PhD in Control Engineering in 2011 from Supelec (now CentraleSupelec), France with an application of set-theoretic methods for fault detection and isolation. His interests are constrained optimization control, set theoretic methods, fault tolerant control, mixed integer programming, motion planning.
\end{IEEEbiography}

\begin{IEEEbiography}[{\includegraphics[width=1in,height=1.25in,clip,keepaspectratio]{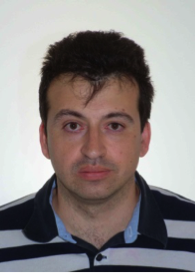}}]{Vicenç Puig}
 holds a PhD in Automatic Control, Vision and Robotics and is the leader of the research group Advanced Control Systems (SAC) at the Polytechnic University of Catalonia.
 He has important scientific contributions in the areas of fault diagnosis and fault tolerant control using interval models. He participated in more than 20 international and national research projects in the last decade. He led many private contracts, and published more than 80 articles in JCR journals and more than 350 in international conference/workshop proceedings. Prof. Puig supervised over 20 PhD theses and over 50 MA/BA theses.
\end{IEEEbiography}

\vfill

\end{document}